\documentclass[a4paper,USenglish]{lipics-v2021}
\usepackage{craigstyle}

\usepackage{misc}
\usepackage{amsmath}
\usepackage{amsfonts}
\usepackage{amssymb}
\usepackage{txfonts}
\usepackage{graphicx}
\usepackage{latexsym}
\usepackage{xspace}
\usepackage{times}
\usepackage[T1]{fontenc} 
\usepackage{tikz}
\usetikzlibrary{automata}
\usetikzlibrary{positioning,arrows,fit,calc}
\usetikzlibrary{decorations.pathmorphing}
\usetikzlibrary{shapes}
\usetikzlibrary{backgrounds}
\usepackage{xspace}

\usepackage{graphicx}
\usepackage{xcolor}

\usepackage{hyperref}
\tikzset{>=latex, 
	point/.style = {circle,draw,thick,minimum size=2mm,inner sep=0pt},
	point1/.style = {circle,draw,thick,minimum size=6mm,inner sep=0pt},
	hm/.style = {dotted,semithick},
	role/.style = {thick},
	tree/.style = {rounded corners=10pt, dashed, fill opacity=0.5, fill=nullscolour},
	wiggly/.style={thick,
	},
	query/.style={thick},
	itria/.style={
  draw,dashed,shape border uses incircle,
  isosceles triangle,shape border rotate=90,yshift=-1.45cm},
  square/.style={regular polygon,regular polygon sides=4}
}

\renewcommand{\L}{L}
\newcommand{\LS}{L'}
\newcommand{\MSO}{\ensuremath{\mathsf{MSO}}}
\newcommand{\LTL}{\ensuremath{\mathsf{LTL}}}
\newcommand{\LTLd}{\ensuremath{\mathsf{LTL}[\Diamond]}}
\newcommand{\LTLnd}{\ensuremath{\mathsf{LTL}[\nxt\Diamond]}}
\newcommand{\CTL}{\ensuremath{\mathsf{CTL}}}
\newcommand{\PDL}{\ensuremath{\mathsf{PDL}}}
\newcommand{\lang}{{\boldsymbol{L}}}

\newcommand{\nxt}{{\ensuremath\raisebox{0.25ex}{\text{\scriptsize$\bigcirc$}}}}
\newcommand{\U}{\mathbin{\mathcal{U}}}

\newcommand{\FO}{\ensuremath{\mathsf{FO}}}
\newcommand{\SO}{\ensuremath{\mathsf{SO}}}
\newcommand{\GSO}{\ensuremath{\mathsf{GSO}}}
\newcommand{\FONE}{\ensuremath{\mathsf{FO_\ne}}}
\newcommand{\GF}{\ensuremath{\mathsf{GF}}}

\newcommand{\FOT}{\ensuremath{\mathsf{FO^2}}}

\newcommand{\CT}{\ensuremath{\mathsf{C^2}}}

\newcommand{\SF}{\ensuremath{\mathsf{S4}}}
\newcommand{\K}{\ensuremath{\mathsf{K}}}
\newcommand{\Alt}{\ensuremath{\mathsf{Alt}}}
\newcommand{\Kn}{\ensuremath{\mathsf{K^{nom}}}}
\newcommand{\KF}{\ensuremath{\mathsf{K4}}}
\newcommand{\SFi}{\ensuremath{\mathsf{S5}}}

\newcommand{\QSF}{\ensuremath{\smash{\mathsf{Q^1S5}}}}

\renewcommand{\ML}{\ensuremath{\mathsf{ML}}}
\newcommand{\muML}{\ensuremath{\mu\mathsf{ML}}}

\newcommand{\MLn}{\ensuremath{\mathsf{ML\!^{nom}}}}

\newcommand{\MLu}{\ensuremath{\mathsf{ML\!^u}}}

\newcommand{\GML}{\ensuremath{\mathsf{GML}}}

\newcommand{\KFT}{\mathsf{K4.3}}
\newcommand{\SFT}{\mathsf{S4.3}}
\newcommand{\wKF}{\mathsf{wK4}}
\renewcommand{\DL}{\mathsf{DL}}

\newcommand{\sig}{\textit{sig}}
\newcommand{\sub}{\textit{sub}}

\newcommand{\frag}{\mathsf{F}}

\renewcommand{\V}{\mathfrak{v}}
\newcommand{\var}{\textit{var}}
\newcommand{\rr}{w}
\newcommand{\bis}{\boldsymbol{\beta}}

\newcommand{\type}{\mathfrak{t}}
\newcommand{\B}{\mathfrak{bt}}
\renewcommand{\T}{\mathfrak{b}}

\renewcommand{\i}{\boldsymbol{i}}
\renewcommand{\j}{\boldsymbol{j}}

\newcommand{\coNP}{\textnormal{\sc coNP}\xspace}

\newcommand{\avec}[1]{\boldsymbol{#1}}

\newcommand{\abc}{A}
\newcommand{\abcr}{A_\varrho}
\newcommand{\sg}{{\boldsymbol{S}}}
\newcommand{\sgo}{\mathop{\cdot}}
\newcommand{\FOo}{\FO(<)}
\newcommand{\ino}[1]{i_{#1}}
\newcommand{\ap}{{\boldsymbol{a}}}
\newcommand{\MSOo}{\MSO(<)}

\newcommand{\satS}{\sg^\dag}
\newcommand{\powerset}[1]{2^{#1}}
\newcommand{\stemp}{\sigma_{\textit{lin}}}

\newcommand{\LTLinf}{\LTL_\omega}

\title{From Interpolating Formulas to Separating Languages and Back Again}

\author{Agi Kurucz}{King's College London, UK}{agi.kurucz@kcl.ac.uk}{0000-0002-6233-6277}{}

\author{Frank Wolter}{University of Liverpool, UK}{wolter@liverpool.ac.uk}{0000-0002-4470-606X}{}

\author{Michael Zakharyaschev}{Birkbeck, University of London, UK}{m.zakharyaschev@bbk.ac.uk}{0000-0002-2210-5183}{}

\begin{document}

\maketitle              

\begin{abstract}
	Traditionally, research on Craig interpolation is concerned with $(a)$ establishing the Craig interpolation property (CIP) of a logic saying that every valid implication in the logic has a Craig interpolant and $(b)$ designing algorithms that extract Craig interpolants from proofs. Logics that lack the CIP are regarded as `pathological' and excluded from consideration. In this chapter, we survey variations and generalisations of traditional Craig interpolation. First, we consider Craig interpolants for implications in  logics without the CIP, focusing on the decidability and complexity of deciding their existence. We then generalise interpolation by looking for Craig interpolants in languages $L'$ that can be weaker than the language $L$ of the given implication. Thus, do not only we restrict the non-logical symbols of Craig interpolants but also the logical ones. The resulting $L/L'$-interpolation problem generalises $L/L'$-definability,
	the question whether an $L$-formula is equivalent to some $L'$-formula. After that, we move from logical languages to formal languages where interpolation disguises itself as separation: given two disjoint languages in a class $\mathcal{C}$, does there exist a separating language in a smaller class $\mathcal{C}'$? This question is particularly well-studied in the case when the input languages are regular and the separating language is first-order definable. Finally, we connect the different research strands by showing how the decidability of the separation problem for regular languages can be used to prove the decidability of Craig interpolant existence for linear temporal logic \LTL.      
\end{abstract}

\tableofcontents


\section{Introduction}\label{intro}

Suppose $L$ is a logic and $\varphi$, $\psi$ are formulas in the language of $L$. 
A \emph{Craig interpolant} for the entailment $\varphi\models_{L} \psi$ in $L$ is a formula $\chi$, all of whose non-logical symbols occur in both $\varphi$ and $\psi$, such that the entailments $\varphi\models_{L}\chi$ and $\chi\models_{L} \psi$ are valid in $L$. The logic $L$ is said to have the \emph{Craig interpolation property} (CIP, for short) if \emph{every} valid entailment \mbox{$\varphi\models_{L} \psi$} has a Craig interpolant $\chi$ in $L$. Since Craig's~\cite{craig_1957} discovery that first-order logic, $\FO$, enjoys the CIP, two questions have traditionally been addressed for every notable logic $L$: Does $L$ have the CIP? And if it does, how to uniformly extract Craig interpolants from formal proofs of entailments $\varphi\models_{L} \psi$? 
Logics lacking interpolants for some valid entailments---the ugly ducklings of the family---have only recently drawn attention in the Craig interpolation line of research and applications.


Our aim in this chapter is to present a broader view on the interpolation-related research directions and discuss most important results. We proceed in three generalisation steps. The first of them breaks the traditional uniform `all or nothing' approach to the CIP and turns it into a \emph{decision problem}. 

\subsection{Craig Interpolant Existence Problem} 

We start by looking at a few important logics $L$ that \emph{do not} enjoy the CIP: propositional modal logics such as $\SFT$, weak $\KF$, and $\Alt_{3}$, modal logics with nominals, decidable fragments of $\FO$ such as the two variable and guarded fragments $\FOT$ and $\GF$, linear temporal logic \LTL, and first-order modal logics under the constant domain semantics.
For all of these logics $L$, the existence of a Craig interpolant for an entailment $\varphi\models_{L} \psi$ does not follow from the validity of that entailment in $L$. We regard it as a new type of decision problem for $L$, the \emph{interpolant existence problem} (IEP): given input formulas $\varphi$ and $\psi$, decide whether the entailment $\varphi\models_{L} \psi$ has a Craig interpolant in $L$. As the IEP for $L$ typically generalises  entailment in $L$, the former might be computationally harder than the latter; in fact, the IEP can turn out to be undecidable for some decidable $L$. 
We discuss recent developments in this area, focusing on the logics listed above.

Craig interpolation is concerned with the shared non-logical symbols of the premise $\varphi$ and conclusion $\psi$. One can think of Craig interpolants as eliminating the non-shared symbols in the premise  and conclusion, thereby providing an `explanation' of the entailment in relevant terms. For example, in the extreme case when the signatures of propositional formulas $\varphi$ and $\psi$ are disjoint, only the logical constants $\bot$ or $\top$ can be a Craig interpolant for $\varphi\models_{L} \psi$, which means that either $\neg\varphi$ or $\psi$ is valid in $L$. 
%
%
To find an `explanation in simpler terms'\!, one could look for interpolants that eliminate certain logical constructs from the language, say recursion or counting. This leads us to the second step of generalisation.

\subsection{Interpolation and Definability in Weaker Languages}

Let $L'$ be a sublanguage of $L$. For example, $L$ could be the modal $\mu$-calculus while $L'$ the basic modal logic (without fixed-point operators). The $L/L'$-\emph{interpolant existence problem} (or $L/L'$-IEP) is to decide whether given $L$-formulas $\varphi$ and $\psi$ have an $L'$-\emph{interpolant}, that is, $\varphi \models_L \chi$ and $\chi\models_L \psi$, for some $L'$-formula $\chi$. 
For this generalisation of Craig interpolation, the existence of an $L/L'$-interpolant typically does not follow from the validity of $\varphi \models_L \psi$, and we are again facing a decision problem that can be much harder than entailment in $L$.

The $L/L'$-IEP also generalises the \emph{$L/L'$-definability problem}, which asks whether a given formula in an expressive logic $L$ is \emph{logically equivalent} to some formula in a weaker language $L'$. 
This definability problem has been considered for numerous pairs $L/L'$ 
of logics including $\FO$ and modal logic (van Benthem's theorem)~\cite{Benthem83}, 
MSO and modal $\mu$-calculus~\cite{DBLP:conf/concur/JaninW96}, various guarded logics with and without recursion~\cite{DBLP:journals/tocl/GradelHO02,DBLP:journals/lmcs/BenediktBB19} as well as datalog and unions of conjunctive queries in database theory~\cite{DBLP:conf/stoc/CosmadakisGKV88,DBLP:journals/jcss/Naughton89}. 
We discuss work on generalised interpolation and definability for modal logic, the modal $\mu$-calculus, and fragments of $\FO$ with and without counting. 

The $L/L'$-IEP can also be viewed as a separation problem for the structures defined by $L$-formulas. If $\varphi \models_L \psi$, then $\varphi$ and $\neg \psi$ define disjoint classes $S_\varphi$ and $S_{\neg\psi}$ of $L$-structures (in formal verification, they could describe, respectively, a set of initial states in a transition system and a set of `bad' states that should not be reached via any computation). Then an $L'$-interpolant $\chi$ for $\varphi \models_L \psi$ defines a class $S_\chi$ of $L'$-structures  separating $S_\varphi$ and $S_{\neg\psi}$ in the sense that $S_\varphi \subseteq S_\chi$ and $S_{\neg\psi} \cap S_\chi = \emptyset$.

\subsection{Defining and Separating Formal Languages}\label{subsec:flanguages}

Our third step consists of looking into definability and separability in the context of formal languages, where
the study of these problems has existed and flourished independently from Craig interpolation for many years.
The \emph{separation} problem asks,
given two 
formal languages $\lang_{1}$ and $\lang_{2}$ from some class $\mathcal{C}$ (say, specified by means of some device such as a finite automaton), whether there exists a language $\lang$ belonging to some subclass $\mathcal{C}'$ (say, specified using a less powerful device, such as a counter-free automaton) that separates $\lang_{1}$ and $\lang_{2}$ in the sense that
$\lang_1\subseteq\lang$ and $\lang_2\cap\lang=\emptyset$.
If $\mathcal{C}$ is closed under complementation, then
separation generalises \emph{definability} (also known as \emph{membership} or \emph{characterisation\/}), the problem that asks whether a given language in $\mathcal{C}$ belongs to $\mathcal{C}'$.

A large chunk of the research in this area has been focused on deciding and separating regular languages. 
In Section~\ref{sec:formal-languages}, we discuss in detail a particular instance of these questions: the decidability of defining/separating regular languages consisting of finite words by star-free (or, equivalently, \FO-definable)  languages.
We also provide further references to work in a wider context including 
languages with infinite and tree-shaped words, employing $\FO$-fragments as definitions/separators, and defining/separating non-regular languages by means of regular ones.


Rather surprisingly, it turns out that this final third step leads to a beautiful solution to a 
purely logical problem encountered in the first step: 
In Section~\ref{sec:LTL}, we show how the Craig interpolant existence problem for propositional linear temporal logic \LTL{} over finite timelines 
(or, equivalently, \FO{} over finite words) 
can be solved by deciding separability of regular languages by star-free languages.

%
%
%


\section{First-Order and Normal Modal Logics: a Quick Reminder}\label{Sec:2}

The purpose of this section is to introduce the (standard) notation we use when considering first-order logic and normal modal logics later on in the chapter. For more details on the basics of those logics the reader is referred to~\cite{modeltheory,DBLP:books/daglib/0030819,Pratt23book}. 

The variant of first-order logic, $\FO$, we consider here is equipped with equality but does not have constant and function symbols. Hence \emph{$\FO$-formulas} are defined by the grammar
\[
\varphi \ \ := \ \ \top \ \mid \ R(x_{1},\dots,x_{n}) \ \mid \  x_{1} = x_{2} \ \mid \ \varphi \wedge \varphi' \ \mid \ \neg \varphi \ \mid \ \exists x \, \varphi ,
\]
where $\top$ is the logical constant `truth'\!, the $x_i$ and $x$ range over individual variables, and $R$ over predicate symbols of arity $n$, for any $n\ge 0$. For a tuple $\avec{x}$ of variables, write $\varphi(\avec{x})$ to indicate that the free variables of $\varphi$ occur in $\avec{x}$. A set $\sigma$ of predicate symbols is called a \emph{signature} (for $\FO$); the signature $\sig(\varphi)$ of a formula $\varphi$ comprises those predicate symbols that occur in $\varphi$. A \emph{$\sigma$-structure} takes the form 
\[
\Amf = \big( \dom(\Amf),(R^{\Amf})_{R\in \sigma} \big), 
\] 
where $\dom(\Amf) \ne \emptyset$ is the \emph{domain} of $\Amf$ and $R^{\Amf}$ is a relation over $\dom(\Amf)$ of the same arity as $R \in \sigma$. For $\varphi(\avec{x})$ with $\avec{x}=x_{1},\dots,x_{n}$ and a tuple $\avec{a}=a_{1},\dots,a_{n}$ with $a_i \in \dom(\Amf)$, we write $\Amf\models \varphi(\avec{a})$ if $\varphi$ is true in $\Amf$ under the assignment $x_{i}\mapsto a_{i}$.  
We write $\varphi \models \psi$ if $\Amf\models\varphi(\avec{a})$ implies $\Amf\models \psi(\avec{a})$ for all $\sigma$-structures $\Amf$ and tuples $\avec{a}$.
Note that $\varphi\models\psi$ iff $\models\varphi\rightarrow\psi$, that is, $\varphi\rightarrow\psi$ is true in all structures. 

We are also going to consider a few standard \emph{$\FO$-fragments} (that is, subsets of $\FO$-formulas). For example, by $\FONE$ we denote the \emph{equality-free fragment} of $\FO$.  A fragment $\frag$ is said to be  \emph{closed under the Booleans} if $\top\in \frag$ and whenever $\varphi,\psi\in \frag$, then  $\varphi \wedge \psi\in \frag$ and $\neg \varphi\in \frag$.

\begin{definition}\label{def:intFO}\em 
	Let $\frag$ be an $\FO$-fragment and $\varphi(\avec{x}),\psi(\avec{x})\in \frag$. A formula $\chi \in \frag$ is a \emph{Craig interpolant for $\varphi$ and $\psi$ in $\frag$} if 
	\begin{itemize}
		\item $\sig(\chi) \subseteq \sig(\varphi)\cap \sig(\psi)$, 
		
		\item $\varphi(\avec{x})\models \chi(\avec{x})$ and $\chi(\avec{x}) \models \psi(\avec{x})$.
	\end{itemize}
	We say that $\frag$ has the \emph{Craig interpolation property} (CIP) if, for all formulas $\varphi(\avec{x}), \psi(\avec{x})\in \frag$, whenever $\varphi\models\psi$, then there exists a Craig interpolant for $\varphi$ and $\psi$ in $\frag$. 
\end{definition}

Another language, $\ML$, we deal with here is that of propositional modal logic with a \emph{possibility operator} $\Diamond$. $\ML$-\emph{formulas} are defined by the grammar
\[
\varphi \ \ := \ \ \top \ \mid \ p \ \mid \ \varphi \wedge \varphi' \  \mid \ \neg \varphi \ \mid \ \Diamond \varphi,
\]
where $p$ ranges over \emph{propositional variables}. (The other Booleans and the \emph{necessity operator} $\Box$ are regarded as abbreviations; e.g., $\Box \varphi$ stands for $\neg \Diamond \neg \varphi$.) In this case, by a \emph{signature} (for $\ML$) we mean any set $\sigma$ of propositional variables; $\sig(\varphi)$ is the set of all propositional  variables occurring in $\varphi$. A (\emph{Kripke}) \emph{model} for $\ML$ takes the form $\mathfrak{M} = (W,R,\V)$ with a nonempty set $W$ of `worlds'\!, a binary \emph{accessibility relation} $R$ on $W$, and a \emph{valuation} $\V$ mapping every propositional variable (in a given signature $\sigma$) to a subset of $W$. The \emph{truth-relation} $\mathfrak{M},u\models \varphi$, for any $\mathfrak{M}$, $u$ and $\varphi$, is defined inductively by taking
\begin{itemize}
	\item $\mathfrak{M},u\models \top$, 
	
	\item $\mathfrak{M},u\models p$ \ iff \ $u \in \V(p)$,
	
	\item $\mathfrak{M},u\models \varphi \land \psi$ \ iff \ $\mathfrak{M},u\models \varphi$ and $\mathfrak{M},u\models \psi$,
	
	\item $\mathfrak{M},u\models \neg\varphi$ \ iff \ $\mathfrak{M},u \not\models \varphi$,
	
	\item $\mathfrak M,u \models \Diamond \varphi$ \ iff \ $\mathfrak M,v \models \varphi$, for some $v \in W$ with $uRv$.
\end{itemize}
Following standard practice, we often regard a Kripke model $\Mmf=(W,R,\V)$ as a $\sigma$-structure $\Amf$ with $\sigma$ containing the binary predicate $R$ and a unary predicate $p$, for each propositional variable $p$, such that $\dom(\Amf)=W$, $R^{\Amf}=R$, and $p^{\Amf}=\V(p)$, for all $p$. Conversely, for such $\sigma$, one can regard any $\sigma$-structure as a Kripke model. In what follows, we freely use $\ML$-formulas to talk about $\sigma$-structures and $\FO$-formulas to talk about Kripke models. Moreover, $\ML$-formulas can be regarded as $\FO$-formulas with two individual variables $\var = \{x,y\}$: the \emph{standard translation}~\cite{Blackburn_Rijke_Venema_2001} of an $\ML$-formula $\varphi$ is an $\FO$-formula $\varphi^{\ast}_x$ with one free variable $x \in \var$ that is defined inductively by taking 
\begin{multline*}
	\top^{\ast}_x  =  (x=x), \ \ p_x^\ast = p(x), \ \ 
	(\neg \varphi)_x^\ast = \neg \varphi_x^\ast, \ \ (\varphi \land \psi)_x^\ast = \varphi_x^\ast \land \psi_x^\ast,  \ \  
	(\Diamond \varphi)_x^\ast = \exists \bar x \, ( R(x,\bar x) \land \varphi_{\bar x}^\ast ), \\ \text{where $\bar x = y$, $\bar y = x$.}
\end{multline*} 
Then $\Mmf,w\models \varphi$ iff $\Mmf\models \varphi_{x}^{\ast}(w)$, for all pointed models $\Mmf,w$ and $\ML$-formulas $\varphi$. Note that $\Diamond$ is a logical symbol in $\ML$, and so does not belong to $\sig(\varphi)$, for any $\ML$-formula $\varphi$. However, the standard translation converts $\Diamond$ to a non-logical symbol, $R$, which may belong to $\sig(\varphi^*_x)$. Fortunately, this discrepancy does not have any effect on the results for modal logics with a \emph{single} $\Diamond$ we consider in this chapter. The situation is different for logics with \emph{multiple} modalities. 


The set of $\ML$-formulas such that $\mathfrak{M},w\models \varphi$, for all models $\mathfrak{M}$ and all worlds \mbox{$w \in W$}, is often called the \emph{basic} (or \emph{minimal}) \emph{normal modal logic} and denoted by $\K$. In general, a \emph{normal modal logic} $L$ is a subset of $\ML$ that contains $\K$ and is closed under the rules of modus ponens, necessitation ($\varphi \in L$ implies $\Box\varphi \in L$), and uniform substitution of arbitrary $\ML$-formulas in place of propositional variables. If $\mathfrak{M},w\models \varphi$, for all $\varphi \in L$ and $w \in W$, we call $\mathfrak M$ a \emph{model} for $L$ and write $\mathfrak{M} \models L$. Every normal modal logic $L \subseteq \ML$ is known (see, e.g.,~\cite{DBLP:books/daglib/0030819}) to be determined by the class of its models in the sense that
\[
\varphi \notin L \quad \text{iff} \quad \text{there exist $\mathfrak M \models L$ and $w \in W$ such that $\mathfrak M,w \not\models \varphi$}.
\]
Given a set $\Gamma \subseteq \ML$, we write $\Gamma\models_{L}\psi$ if $\mathfrak{M},w\models \Gamma$ implies $\mathfrak{M}, w \models \psi$, for all $\mathfrak M \models L$ and $w \in W$. Instead of $\{ \varphi \} \models_{L}\psi$, we write simply $\varphi \models_{L}\psi$. The consequence relation $\models_{L}$ is known to be \emph{compact}: if $\Gamma\models_{L}\varphi$, then there exists a finite $\Gamma'\subseteq \Gamma$ with $\Gamma'\models_{L}\varphi$. Note also that $\varphi\models_{L}\psi$ iff $\varphi\rightarrow\psi\in L$. 

The definition of Craig interpolants and CIP for a normal modal logic $L$ follows Definition~\ref{def:intFO} for $\FO$-fragments with $\models$ replaced by $\models_{L}$:\footnote{There are other meaningful definitions of Craig interpolants; see~\refchapter{chapter:nonclassical} and, e.g.,~\cite{GabMaks}. In particular, one can also consider interpolants for the global consequence relation.}

\begin{definition}\em 
Let $L$ be a normal modal logic and $\varphi,\psi\in \ML$. An $\ML$-formula $\chi$ is an \emph{$L$-interpolant for $\varphi$ and $\psi$} if 
$\sig(\chi) \subseteq \sig(\varphi)\cap \sig(\psi)$, $\varphi \models_L \chi$, and $\chi \models_L \psi$. We say that $L$ has the CIP if, for all formulas $\varphi, \psi\in \ML$, whenever $\varphi \models_L \psi$, then there is an $L$-interpolant for $\varphi$ and $\psi$. 
\end{definition}

Later on in this chapter, we also consider the extensions of $\FO$ and $\ML$ with constructs for counting (counting quantifiers $\exists^{\ge k}x$ and graded modalities $\Diamond^{\ge k}$) and recursion (second-order quantifiers over unary predicates and least fixed point operators).


\section{Interpolant Existence in Logics without Craig Interpolation Property}\label{Sec:3}

If a logic $L$ has the CIP, the existence of a Craig interpolant  for formulas $\varphi$ and $\psi$ reduces in constant time to deciding whether $\varphi\models_{L}\psi$ holds (or, equivalently, whether the implication \mbox{$\varphi \to \psi$} is valid in $L$). Although many standard logics do enjoy the CIP, there are countless other logics lacking it, for which this reduction to validity does not work. 
Of course, this does not mean that Craig interpolants in the latter category are of no interest or do not exist in the majority of interesting cases.\footnote{We refer the reader to~\refchapter{chapter:kr} for a discussion of applications of Craig interpolants for logics without CIP in knowledge representation and reasoning.} Hence we are facing the following decision problem:

\begin{description}
\item[\emph{interpolant existence problem \textup{(}IEP\textup{)} for $L$}:] given $L$-formulas $\varphi$ and $\psi$, decide whether $\varphi$ and $\psi$ have a Craig interpolant in $L$. 
\end{description}

\noindent
The IEP is meaningful only for those $L$ that do not have the CIP. It has so far been investigated for the following logics $L$ (see Table~\ref{TableCraig} for a summary of the obtained results on the computational complexity of the IEP for $L$): 
\begin{itemize}
\item propositional modal logics determined by linear frames such as $\SFT$ and $\KFT$~\cite{DBLP:journals/corr/abs-2312-05929}, the logic $\wKF$ (weak $\KF$) of derivative spaces~\cite{DBLP:conf/aiml/KuruczWZ24} and $\Alt_{n}$ of $n$-ary trees (for $n\geq 3$)~\cite{Jeanstacs}, whose lack of the CIP was shown in~\cite{Karpenko&Maksimova2010,GabMaks,DBLP:journals/jphil/Wolter97,Jeanstacs};

\item propositional linear temporal logic \LTL, to be discussed in Section~\ref{sec:LTL}, whose lack of the CIP was observed in  \cite{Maksimova91}; 

\item modal logics with nominals and more expressive decidable hybrid logics~\cite{DBLP:journals/tocl/ArtaleJMOW23}, which do not have the CIP as shown in~\cite{DBLP:journals/jsyml/Cate05,DBLP:series/lncs/KonevLWW09};

\item decidable fragments of $\FO$ such as the two variable fragment $\FOT$ and the guarded fragment $\GF$~\cite{DBLP:conf/lics/JungW21}, which do not have the CIP~\cite{comer1969,Pigozzi71,DBLP:journals/ndjfl/MarxA98,NemetiBeth2,Andreka1,DBLP:conf/lpar/HooglandMO99};

\item decidable fragments of first-order modal logics with constant domains such as $\QSF$, the one-variable fragment of quantified $\SFi$ \cite{DBLP:conf/kr/KuruczWZ23}, whose lack of CIP was shown in~\cite{DBLP:journals/jsyml/Fine79}. 
\end{itemize}

\begin{table}[thp]
\caption{Complexity of validity and the IEP for some modal and first-order logics.}
\begin{center}
	\begin{tabular}{l|c|c|c}
		logic (without the CIP) & validity & IEP lower bound & IEP upper bound\\\hline
		finitely axiomatisable & & \\
		extensions of \ $\KFT$ & $\coNP$ & $\coNP$ & $\coNP$\\\hline
		$\wKF$ & \PSpace & \coNExpTime & \textsc{coN3ExpTime} \\\hline
		$\Alt_n$, for $n \ge 3$ & \PSpace & \coNExpTime & \coNExpTime \\\hline
		$\LTL$ & \PSpace & $\PSpace$ & $4\ExpTime$ \\\hline
		$\Kn$ & \PSpace & $\coNExpTime$ & $\coNExpTime$\\\hline
		$\FOT$ & $\coNExpTime$ & $2\ExpTime$ & $\textsc{coN2ExpTime}$ \\\hline
		$\GF$ & $2\ExpTime$ & $3\ExpTime$ & $3\ExpTime$\\\hline
		$\QSF$ & $\coNExpTime$ & $2\ExpTime$ & $\textsc{coN2ExpTime}$
	\end{tabular}
\end{center}
\label{TableCraig}
\end{table}

In this section, we give a brief introduction to model-theoretic approaches to the IEP using characterisations of logical indistinguishability of models by means of appropriate bisimulations. We focus on propositional modal logics with and without nominals. Algebraic approaches used for linear temporal logic \LTL{} are discussed in Sections~\ref{sec:formal-languages} and~\ref{sec:LTL}.


\subsection{Interpolant Existence in Normal Modal Logics}\label{nml}

There are a continuum of normal modal logics; continuum-many of them enjoy the CIP and as many lack it. No algorithm can decide whether a finitely axiomatisable modal logic has the CIP, even among extensions of the modal logic $\KF$ (determined by transitive models). On the other hand, there are at most $36$ logics with the CIP out of a continuum of extensions of $\SF$ (determined by transitive and reflexive models), and they can be effectively  recognised given the axioms of a logic. The reader can find these and other related results in~\refchapter{chapter:nonclassical} and ~\cite{GabMaks,DBLP:books/daglib/0030819}. Here, our concern is deciding the IEP for logics without the CIP.

Denote by $\ML(\sigma)$ the set of $\ML$-formulas in a signature $\sigma$, and let $\mathfrak{M}_{i} = (W_i,R_i,\V_i)$, $i=1,2$, be two models with some $w_{i}\in W_{i}$. The pointed models $\mathfrak{M}_{1},w_{1}$ and $\mathfrak{M}_{2},w_{2}$ are called \emph{$\sigma$-indistinguishable in $\ML$} when $\mathfrak{M}_{1}, w_{1} \models\varphi$ iff $\mathfrak{M}_{2},w_{2} \models \varphi$, for all $\varphi \in \ML(\sigma)$, in which case we write 
$\Mmf_{1},w_{1}\equiv_{\ML(\sigma)} \Mmf_{2},w_{2}$. This notion gives the following criterion of the IEP for normal modal logics:

\begin{theorem}\label{thm:critindistmodal}
For any normal modal logic $L$ and any $\ML$-formulas $\varphi$ and $\psi$, the following conditions are equivalent\textup{:}
\begin{itemize}
	\item there does not exist an $L$-interpolant for $\varphi$ and $\psi$\textup{;}
	
	\item there exist two pointed models $\Mmf_{i},w_i$, $i=1,2$, for $L$ such that $\Mmf_{1},w_{1}\models \varphi$, $\Mmf_{2},w_{2}\models\neg\psi$ and $\Mmf_{1},w_{1}\equiv_{\ML(\sigma)} \Mmf_{2},w_{2}$, where $\sigma=\sig(\varphi)\cap \sig(\psi)$. 
\end{itemize}
In particular, one can take the $\Mmf_i$ to be the canonical model for $L$~{\rm \cite{DBLP:books/daglib/0030819,Blackburn_Rijke_Venema_2001}}.
\end{theorem}
\begin{proof}
$(\Leftarrow)$ Let $\Mmf_{1},w_{1}\models \varphi$, $\Mmf_{2},w_{2}\models \neg\psi$ and $\Mmf_{1},w_{1} \equiv_{\ML(\sigma)} \Mmf_{2},w_{2}$. Suppose $\chi$ is an $L$-interpolant for $\varphi$ and $\psi$. Then $\varphi\models_{L}\chi$, and so $\Mmf_{1},w_{1}\models \chi$. As $\sig(\chi)\subseteq \sigma$ and the $\Mmf_{i},w_{i}$ are $\sigma$-indistinguishable, we have $\Mmf_{2},w_{2}\models \chi$, and so $\Mmf_{2},w_{2}\models \psi$, which is a contradiction. It follows that $\varphi$ and $\psi$ do not have an $L$-interpolant.

$(\Rightarrow)$ Conversely, suppose there is no $L$-interpolant for $\varphi$ and $\psi$. Consider the set
\[
\Phi^{\sigma}= \{ \chi \in \ML(\sigma) \mid \varphi\models_{L} \chi\}.
\]
By compactness, we have $\Phi^{\sigma} \not\models_{L} \psi$, for otherwise there would be a finite $\Phi \subseteq \Phi^{\sigma}$ with $\Phi \models_{L} \psi$, and so $\bigwedge \Phi$ would be an $L$-interpolant for $\varphi$ and $\psi$. Then there is a model $\Mmf_{2}$ (say, the canonical model) for $L$ such that $\Mmf_{2},w_{2} \models \Phi^{\sigma}\cup \{\neg\psi\}$, for some $w_2 \in W_2$. Let 
\[
\type_{\Mmf_{2}}^{\sigma}(w_2) = \{ \chi \in \ML(\sigma) \mid  \Mmf_{2},w_{2}\models \chi\}.
\]
Using compactness again, we find a model $\Mmf_{1}$ for $L$ with $\Mmf_{1},w_{1} \models \type_{\Mmf_{2}}^{\sigma}(w_2) \cup \{\varphi\}$, for some $w_1 \in W_1$ (because otherwise we would have a finite set $\Psi \subseteq \type_{\Mmf_{2}}^{\sigma}(w_2)$ with $\varphi \models_{L} \neg \bigwedge \Psi$, which means that $\neg\bigwedge\Psi \in \Phi^{\sigma} \subseteq \type_{\Mmf_{2}}^{\sigma}(w_2)$, and so $\type_{\Mmf_{2}}^{\sigma}(w_2)$ would be unsatisfiable). Thus, we have $\Mmf_{1},w_{1}\models \varphi$, $\Mmf_{2},w_{2}\models \neg\psi$, and $\Mmf_{1},w_{1} \equiv_{\ML(\sigma)} \Mmf_{2},w_{2}$. \end{proof}

The criterion of Theorem~\ref{thm:critindistmodal} is difficult to use for deciding the IEP as it quantifies over infinitely many $\sigma$-formulas. One can make it more `operational' by replacing $\equiv_{\ML(\sigma)}$ with a model-theoretic characterisation in terms of $\sigma$-bisimulations to be introduced below. 

Let again $\mathfrak{M}_{i} = (W_i,R_i,\V_i)$, $i=1,2$, be models for $L$ and let $\sigma$ be a signature. 
A relation $\bis \subseteq W_1 \times W_2$ is called a $\sigma$-\emph{bisimulation} between $\mathfrak M_1$ and $\mathfrak M_2$ if the following conditions are satisfied  whenever $u_1 \bis u_2$:
\begin{itemize}
\item 
$\mathfrak{M}_1,u_1 \models p$ iff $\mathfrak{M}_2,u_2 \models p$, for all propositional variables $p \in \sigma$\textup{;}

\item if $u_1R_1v_1$, then there is $v_2$ such that $u_2R_2v_2$ and $v_1 \bis v_2$; and, conversely, if $u_2R_2 v_2$, then there is $v_1$ with $u_1R_1 v_1$ and $v_1 \bis v_2$.
\end{itemize}
If there is such a $\sigma$-bisimulation $\bis$ with $w_1\bis w_2$, we write $\mathfrak{M}_1,w_1 \sim_{\ML(\sigma)} \mathfrak{M}_2,w_2$. A remarkable property of the canonical models (and also of $\omega$-saturated $\sigma$-structures) $\Mmf_L$ for $L$ is that $\mathfrak{M}_L,u \equiv_{\ML(\sigma)} \mathfrak{M}_L,v$ iff $\mathfrak{M}_L,u \sim_{\ML(\sigma)} \mathfrak{M}_L,v$, for any $u$ and $v$ in $\Mmf$. This yields the operational criterion we are after, see also~\refchapter{chapter:modal} and~\refchapter{chapter:kr}:

\begin{theorem}\label{thm:critindistmodalbisim}
For any normal modal logic $L$ and any $\ML$-formulas $\varphi$ and $\psi$, the following conditions are equivalent\textup{:}
\begin{itemize}
	\item there does not exist an $L$-interpolant for $\varphi$ and $\psi$\textup{;}
	
	\item there exist two pointed models $\Mmf_{i},w_i$, $i=1,2$, for $L$ such that $\Mmf_{1},w_{1}\models \varphi$, $\Mmf_{2},w_{2}\models\neg\psi$ and $\Mmf_{1},w_{1} \sim_{\ML(\sigma)} \Mmf_{2},w_{2}$, where $\sigma=\sig(\varphi)\cap \sig(\psi)$. 
\end{itemize}
In particular, one can take the $\Mmf_i$ to be the canonical model for $L$.
\end{theorem}

We illustrate this criterion by showing that normal modal logics $\wKF$, $\SFT$, and $\Alt_3$ do not have the CIP~\cite{GabMaks,Karpenko&Maksimova2010}. The first of them is called \emph{weak} $\KF$~\cite{Esakia2001} as it is characterised by \emph{weakly transitive} models satisfying   
the condition $\forall x,y,z \in W\, \big( R(x,y) \land R(y,z) \to (x=z) \vee R(x,z) \big)$; see also recent~\cite{DBLP:journals/jacm/BaltagBF23}. 
$\SFT$ is characterised by models $\mathfrak{M} = (W,R,\V)$, in which $R$ is transitive, reflexive, and \emph{weakly connected} in the sense that
\[
\forall x,y,z \in W \, \big(R(x,y)\land R(x,z) \to R(y,z) \lor R(z,y)\big).
\]
$\Alt_n$, for $n \ge 1$, is the logic characterised by models whose underlying frames $(W,R)$ are (irreflexive and intransitive) directed trees of outdegree $n$. 

In the pictures below, $\bullet$ denotes an irreflexive point, $\circ$ a reflexive one, and an ellipse represents a \emph{cluster}, i.e., a set of points where any two distinct ones `see' each other via the accessibility relation.

\begin{example}\label{ex:basic} 
$(i)$ Suppose 
$\varphi = \Diamond\Diamond p \land \neg \Diamond p$ and $\psi = \Diamond \Diamond \neg q \lor q$ 
with the common signature $\sigma = \emptyset$. 
It is easy to see that $\varphi \to \psi$ is true in all weakly transitive models, and so $\varphi \models_{\wKF} \psi$. 
On the other hand, the models $\mathfrak M_\varphi$ and $\mathfrak M_\psi$ depicted below are such that $\mathfrak M_\varphi,\rr_\varphi \models \varphi$, $\mathfrak M_\psi,\rr_\psi \models \neg \psi$, and 
the \emph{universal relation} $\bis$ between the points of $\mathfrak M_\varphi$ and $\mathfrak M_\psi$ is a $\sigma$-bisimulation
with $\rr_\varphi \bis \rr_\psi$: 
\\[3pt]
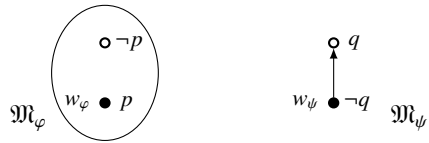
\centerline{\begin{tikzpicture}[>=latex,line width=0.4pt,xscale = 1,yscale = 1]
		\node at (-1,-.6) {$\mathfrak M_\varphi$};
		\draw[] (0,0) ellipse (.7 and .9);
		\node[point,scale = 0.6,label=right:\!{\footnotesize $\neg p$}] (f1) at (0,.4) {};
		\node[point,fill=black,scale = 0.6,label=left:{\footnotesize $\rr_\varphi$},label=right:{\footnotesize $p$}] (f0) at (0,-.4) {};
		\node at (4,-.6) {$\mathfrak M_\psi$};
		\node[point,scale = 0.6,label=right:{\footnotesize $q$}] (p1) at (3,.4) {};
		\node[point,fill=black,scale = 0.6,label=left:{\footnotesize $\rr_\psi$},label=right:\!{\footnotesize $\neg q$}] (p0) at (3,-.4) {};
		\draw[->] (p0) to (p1);
	\end{tikzpicture}
}\\[3pt]
Therefore, $\varphi$ and $\psi$ do not have a $\wKF$-interpolant, and so $\wKF$ does not enjoy the CIP.

$(ii)$ Next, consider the formulas 
%
\[
\varphi = \Diamond( p_{1} \wedge  \Diamond \neg q_{1}) \wedge \Box (p_{2} \rightarrow \Box q_{1}) 
,\quad  \psi = \neg [ \Diamond( p_{2} \wedge \Diamond\neg q_{2}) \land \Box (p_{1} \rightarrow \Box q_{2}) ]
\]
with the common signature $\sigma = \{p_1,p_2\}$. 
It is not hard to see that $\varphi \models_{\SFT} \psi$ and that, in the picture below, $\mathfrak{M}_\varphi,\rr_\varphi \models \varphi$, $\mathfrak{M}_\psi,\rr_\psi \models \neg\psi$, and the relation $\bis$ shown by the dotted lines is a $\sigma$-bisimulation 
between $\mathfrak M_\varphi$ and $\mathfrak M_\psi$ with $\rr_\varphi \bis \rr_\psi$.\\
\centerline{
	\begin{tikzpicture}[>=latex,line width=0.5pt,xscale = 1.15,yscale = 1]
		\node[]  at (-1,0) {$\mathfrak M_\psi$};
		\node[point,scale = 0.7,label=below:{\footnotesize $\neg\psi$},label=above:{\footnotesize $\qquad \rr_\psi$}] (x2) at (0,0) {};
		\node[point,scale = 0.7,label=below:{\footnotesize $p_2,\neg q_2$}] (y2) at (1,0) {};
		\node[point,scale = 0.7,label=above:{\footnotesize $p_2,q_2$}] (a20) at (2.5,0) {};
		\node[point,scale = 0.7,label=above:{\footnotesize $p_1,q_2$}] (a21) at (3.5,0) {};
		\draw[] (3,.1) ellipse (1 and .75);
		\draw[->] (x2) to (y2);
		\draw[->] (y2) to (2,0);
		%
		\node[]  at (-1,3) {$\mathfrak M_\varphi$};
		\node[point,scale = 0.7,label=above:{\footnotesize $\varphi$},label=below:{\footnotesize $\qquad \rr_\varphi$}] (x1) at (0,3) {};
		\node[point,scale = 0.7,label=above:{\footnotesize $p_1,\neg q_1$}] (y1) at (1,3) {};
		\node[point,scale = 0.7,label=,label=below:{\footnotesize $p_2,q_1$}] (a10) at (2.5,3) {};
		\node[point,scale = 0.7,label=below:{\footnotesize $p_1,q_1$}] (a11) at (3.5,3) {};
		\draw[] (3,2.9) ellipse (1 and .75);
		\draw[->] (x1) to (y1);
		\draw[->] (y1) to (2,3);
		%
		\node[]  at (-.5,1.5) {$\bis$};
		\draw[thick,dotted] (x1) to (x2);
		\draw[thick,dotted] (y1) to (a21);
		\draw[thick,dotted] (y2) to (a10);
		\draw[thick,dotted] (a10) to (a20);
		\draw[thick,dotted] (a11) to (a21);
	\end{tikzpicture}
}
Thus, $\varphi$ and $\psi$ do not have an $\SFT$-interpolant, and so $\SFT$ does not enjoy the CIP.

$(iii)$ Finally, suppose $\varphi = \Diamond (p \land q) \land \Diamond (p \land \neg q)$ and $\psi = \neg [ \Diamond (\neg p \land r) \land \Diamond (\neg p \land \neg r)]$ with common signature $\sigma = \{p\}$. The reader can readily check that $\varphi \to \psi$ is true in all models based on trees of outdegree 3, and so $\varphi \models_{\Alt_3} \psi$. On the other hand, in the picture below, $\mathfrak{M}_\varphi,\rr_\varphi \models \varphi$, $\mathfrak{M}_\psi,\rr_\psi \models \neg\psi$, and the relation $\bis$ shown by the dotted lines is a $\sigma$-bisimulation 
between $\mathfrak M_\varphi$ and $\mathfrak M_\psi$ with $\rr_\varphi \bis \rr_\psi$. It follows that $\varphi$ and $\psi$ do not have an interpolant in $\Alt_3$. \lipicsEnd
\end{example}

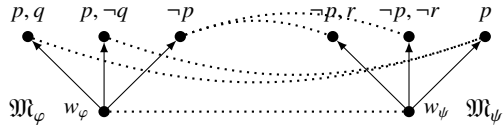
\centerline{\begin{tikzpicture}[>=latex,line width=0.4pt,xscale = 1,yscale = 1]
	\node at (-1,0) {$\mathfrak M_\varphi$};
	\node[point,fill=black,scale = 0.6,label=left:{\footnotesize $\rr_\varphi$}] (f0) at (0,0) {};
	\node[point,fill=black,scale = 0.6,label=above:{\footnotesize $p,q$}] (f1) at (-1,1) {};
	\node[point,fill=black,scale = 0.6,label=above:{\footnotesize $p,\neg q$}] (f2) at (0,1) {};
	\node[point,fill=black,scale = 0.6,label=above:{\footnotesize $\neg p$}] (f3) at (1,1) {};
	\draw[->] (f0) to (f1);
	\draw[->] (f0) to (f2);
	\draw[->] (f0) to (f3);
	\node at (5,0) {$\mathfrak M_\psi$};
	\node[point,fill=black,scale = 0.6,label=right:{\footnotesize $\rr_\psi$}] (g0) at (4,0) {};
	\node[point,fill=black,scale = 0.6,label=above:{\footnotesize $\neg p,r$}] (g1) at (3,1) {};
	\node[point,fill=black,scale = 0.6,label=above:{\footnotesize $\neg p,\neg r$}] (g2) at (4,1) {};
	\node[point,fill=black,scale = 0.6,label=above:{\footnotesize $p$}] (g3) at (5,1) {};
	\draw[->] (g0) to (g1);
	\draw[->] (g0) to (g2);
	\draw[->] (g0) to (g3);
	\draw[thick,dotted] (f0) to (g0); 
	\draw[thick,dotted,bend right=20] (f1) to (g3); 
	\draw[thick,dotted,bend right=20] (f2) to (g3); 
	\draw[thick,dotted,bend left=20] (f3) to (g1); 
	\draw[thick,dotted,bend left=20] (f3) to (g2); 
\end{tikzpicture}
}

Theorem~\ref{thm:critindistmodalbisim} can be used as a starting point for establishing  decidability of the IEP for normal modal logics lacking the CIP and finding its computational complexity. 
A standard way of proving decidability of a finitely axiomatisable modal logic $L$ is by showing (if lucky) that it has the \emph{finite model property} (FMP): $\varphi \notin L$ iff there is a finite pointed model $\Mmf,w$ for $L$ with $\Mmf,w \not \models \varphi$; see, e.g.,~\cite{DBLP:books/daglib/0030819}. Likewise, to prove decidability of the IEP for such an $L$, one can try to show that it has the \emph{finite bisimilar model property} (FBMP) in the sense that $\varphi$ and $\psi$ do not have an $L$-interpolant iff  this is witnessed by \emph{finite} models $\Mmf_i,w_i$ in the second item of Theorem~\ref{thm:critindistmodalbisim}. We give an example of such a proof in Section~\ref{sec:nominals}, Proposition~\ref{prop:FBMP}. Here, we only observe the following obvious analogue of Harrop's theorem~\cite[Theorem 16.13]{DBLP:books/daglib/0030819}:

\begin{lemma}
If a finitely axiomatisable normal modal logic $L$ has the FBMP, then the IEP for $L$ is  decidable.
\end{lemma}

We would also obtain an upper complexity bound on the IEP for $L$ if the size of the witnessing models $\Mmf_i$ could be bounded by some function in the size $|\varphi| + |\psi|$ of $\varphi$, $\psi$.

Establishing the FBMP of $L$ cannot be easier than showing the FMP. We start by considering any models $\Mmf_i,w_i$ from the second item of Theorem~\ref{thm:critindistmodalbisim} and aim to convert them into finite models $\Mmf'_i,w'_i$ meeting the same conditions. We can try and apply various filtration or maximal point techniques to each of the $\Mmf_i$, but have to coordinate the construction of the new models by keeping track of $\sigma$-bisimilarity, which requires more complex `data structures'\!. For instance, one can show that $\SFT$ has the polynomial FBMP using the maximal point technique~\cite{DBLP:journals/corr/abs-2312-05929}, which implies that the IEP for $\SFT$ is $\coNP$-complete---as complex as the validity problem. On the other hand, using a quite elaborate filtration argument, \cite{DBLP:conf/aiml/KuruczWZ24} establishes a triple-exponential FBMP of $\wKF$, which has the exponential FMP and is $\PSpace$-complete. The next theorem shows that the IEP for $\wKF$ and $\Alt_n$, $n \ge 3$, is harder than validity; cf.\ Table~\ref{TableCraig}.

\begin{theorem}[\cite{DBLP:journals/corr/abs-2312-05929,DBLP:conf/aiml/KuruczWZ24,Jeanstacs}]\label{thm:KF3}

\subtheorem{thm:KF3-1} The IEP for $\wKF$ is decidable in \textsc{coN3ExpTime} while being \coNExpTime-hard. 

\subtheorem{thm:KF3-2} The IEP for any finitely axiomatisable logic containing $\KFT$ is \coNP-complete. 

\subtheorem{thm:KF3-3} The IEP for any $\Alt_n$ with $n\ge 3$ is \coNExpTime-complete. 
\end{theorem}

In Theorem~\ref{thm:KF3} $(ii)$, $\KFT$ is the logic determined by weakly connected transitive models. Not all of its (continuum-many) extensions have the FBMP: a notable example is the logic determined by models based on finite strict linear orders, which has the FMP by definition. The decidability and $\coNP$ complexity upper bound are obtained by showing that witness models $\Mmf_i,w_i$ for Theorem~\ref{thm:critindistmodalbisim} can be assembled from polynomially-many, possibly infinite, blocks with a finitely describable structure.


\subsection{The Finite Bisimilar Model Property of Modal Logic with Nominals}\label{sec:nominals}

To convey the flavour of a filtration-based construction showing the FBMP, in this section we consider the basic modal logic $\Kn$ with nominals. Its formulas are given in the language $\MLn$ extending $\ML$ with another type of atomic formulas, \emph{nominals}, that are denoted by $\i$, $\j$ and interpreted in $\MLn$-models $\Mmf=(W,R,\V)$ by singleton sets, i.e., $|\V(\i)| =1$, for every nominal $\i$. Nominals are regarded as \emph{non-logical} symbols, so they are included in $\sig(\varphi)$ whenever they occur in an $\MLn$-formula $\varphi$. The definition of a $\sigma$-\emph{bisimulation} $\bis$ between $\MLn$-models $\mathfrak M_1$ and $\mathfrak M_2$ requires an additional item for nominals:
\begin{itemize}
\item 
whenever $u_1 \bis u_2$, then $\mathfrak{M}_1,u_1 \models \i$ iff $\mathfrak{M}_2,u_2 \models \i$, for all nominals $\i \in \sigma$.
\end{itemize}
We write $\mathfrak{M}_1,w_1 \sim_{\MLn(\sigma)} \mathfrak{M}_2,w_2$  to say that there is a $\sigma$-bisimulation $\bis$ between $\MLn$-models $\mathfrak M_1$ and $\mathfrak M_2$ with $w_1\bis w_2$. 
We denote by $\Kn$ the basic modal logic with nominals, which comprises those $\MLn$-formulas that are true in all worlds of all models. The decision problem for $\Kn$ is known to be $\PSpace$-complete~\cite{DBLP:conf/csl/ArecesBM99}. Theorems~\ref{thm:critindistmodal} and~\ref{thm:critindistmodalbisim} easily generalise to the case $L = \Kn$ and $\MLn$ in place of $\ML$.

\begin{proposition}\label{nominalnocraig}
$\Kn$ does not have the CIP.
\end{proposition}
\begin{proof}
Consider $\varphi = \i \wedge \Diamond \i$ and $\psi = \j \rightarrow \Diamond \j$ with distinct nominals $\i$ and $\j$, so $\sigma = \sig(\varphi) \cap \sig(\psi) = \emptyset$. As $\Mmf,u \models \varphi$ means that $u$ is reflexive, $\varphi \models_{\Kn} \psi$. Let $\Mmf_\varphi$ consist of one reflexive world $\rr_\varphi$, where $\i$ is true, and let $\Mmf_\psi$ consist of two irreflexive worlds seeing each other, and in one of them, $\rr_\psi$, $\j$ is true:\\
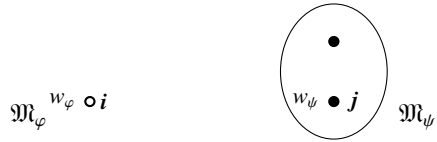
\centerline{\begin{tikzpicture}[>=latex,line width=0.4pt,xscale = 1,yscale = 1]
		\node at (-4,-.6) {$\mathfrak M_\varphi$};
		\node[point,scale = 0.6,label=left:{\footnotesize $\rr_\varphi$},label=right:\!{\footnotesize $\i$}] (p0) at (-3.2,-.4) {};
		\node at (1.1,-.6) {$\mathfrak M_\psi$};
		\draw[] (0,0) ellipse (.7 and .9);
		\node[point,fill=black,scale = 0.6] (f1) at (0,.4) {};
		\node[point,fill=black,scale = 0.6,label=left:{\footnotesize $\rr_\psi$},label=right:{\footnotesize $\j$}] (f0) at (0,-.4) {};
	\end{tikzpicture}
}\\
The universal relation between these two models is a $\sigma$-bisimulation, $\Mmf_\varphi, \rr_\varphi \models \varphi$ and $\Mmf_\psi, \rr_\psi \models \neg\psi$. Therefore, by Theorem~\ref{thm:critindistmodalbisim}, $\varphi$ and $\psi$ do not have a Craig interpolant in $\Kn$.
\end{proof}

\begin{proposition}\label{prop:FBMP}
$\Kn$ has the FBMP.
\end{proposition}
\begin{proof}
We first remind the reader of a standard filtration proof showing that $\Kn$ has the FMP and then generalise it to a proof of the FBMP. Let $\sub(\varphi)$ be the set of all subformulas of an $\MLn$-formula $\varphi$ together with their negations. By the $\varphi$-\emph{type} of a world $w$ in a model $\Mmf = (W,R,\V)$ we mean the set
\[
\type_{\Mmf}(w)= \{ \psi \in \sub(\varphi) \mid \Mmf,w\models \psi\}.
\]
Suppose $\Mmf,w\models \varphi$.
Define a new finite model $\Mmf'=(W',R',\V')$ by taking 
\begin{itemize}
	\item $W'= \{\type_{\Mmf}(u) \mid u\in W\}$;
	
	\item $\type_{\Mmf}(u)\in \V'(p)$ iff $p\in \type_{\Mmf}(u)$, for all propositional variables $p$ and all $u\in W$; 
	
	\item $\type_{\Mmf}(u) \in \V'(\i)$ iff $\i\in \type_{\Mmf}(u)$, for all nominals $\i$ and all $u\in W$;
	
	\item $(\type,\type')\in R'$ iff there are $u,v \in W$ with $(u,v) \in R$, $\type_{\Mmf}(u)=\type$ and $\type_{\Mmf}(v)=\type'$.
\end{itemize}
One can show by induction on $\varphi$ that $\Mmf',\type_{\Mmf}(u) \models \bigwedge \type_{\Mmf}(u)$ for all $u \in W$, and so $\Mmf',\type_{\Mmf}(w) \models \varphi$. 

To generalise this filtration construction to $\sigma$-bisimilar models, suppose $\sigma = \sig(\varphi_1) \cap \sig(\varphi_2)$ and we have  $\Mmf_{1},w_{1}\sim_{\MLn(\sigma)}\Mmf_{2},w_{2}$ and $\Mmf_{l},w_{l}\models \varphi_{l}$, $l=1,2$. Our aim is to convert the $\Mmf_l$ into finite models, denoted $\Mmf'_l$, with the same properties. We define the $(\varphi_1,\varphi_2)$-\emph{bisimulation type} of $u\in W_{l}$, $l=1,2$, as the triple $\B_{\Mmf_{l}}(u)$ defined by taking
\begin{align*}
	& \B_{\Mmf_{l}}(u) = \big( \type_{\Mmf_{l}}(u),\T_{\Mmf_{1}}(u),\T_{\Mmf_{2}}(u) \big),\\
	& \type_{\Mmf_l}(u)= \{ \psi \in \sub(\varphi_1) \cup \sub(\varphi_2) \mid \Mmf_l,u \models \psi\},  \\ 
	& \T_{\Mmf_{k}}(u) = \{\type_{\Mmf_{k}}(v) \mid \Mmf_{l},u \sim_{\MLn(\sigma)} \Mmf_{k},v\}, \ \text{ for } k \in \{1,2\}.
\end{align*}
Note that now the type $\type_{\Mmf_l}(u)$ consists of subformulas of both $\varphi_1$ and $\varphi_2$. Define finite models $\Mmf_{l}'=(W_{l}',R_{l}',\V_{l}')$, $l = 1,2$, by taking 
\begin{itemize}
	\item $W_{l}'= \{\B_{\Mmf_{l}}(u) \mid u\in W_{l}\}$;
	
	\item $\B_{\Mmf_{l}}(u) \in \V'_l(p)$ iff $p \in \type_{\Mmf_{l}}(u)$, for all propositional variables $p$ and all $u\in W_{l}$; 
	
	\item $\B_{\Mmf_{l}}(u) \in \V'_l(\i)$ iff $\i \in \type_{\Mmf_{l}}(u)$, for all nominals $\i$ and all $u\in W_{l}$; 
	
	\item $(\B, \B')\in R_{l}'$ iff there exist worlds $u,v \in W_l$ such that $(u,v) \in R_l$, $\B_{\Mmf_{l}}(u) = \B$ and $\B_{\Mmf_{l}}(v) = \B'$.  
\end{itemize}
Then the following hold: 

\begin{lemma}\label{new}
	\sublemma{new1}
	$\Mmf'_l, \B_{\Mmf_{l}}(u) \models \bigwedge \type_{\Mmf_{l}}(u)$, for all $u\in W_{l}$ with $l=1,2$.
	
	\sublemma{new2} 
	$\Mmf'_1, (\type, \T_1, \T_2)  \sim_{\MLn(\sigma)} \Mmf'_2, (\type', \T_1, \T_2)$, for all $(\type, \T_1, \T_2)\in W_{1}'$ and $(\type', \T_1, \T_2)\in W_{2}'$.
\end{lemma}
\begin{proof}
	The first claim is shown by proving inductively that $\Mmf'_l, \B_{\Mmf_{l}}(u)\models \chi$ iff $\Mmf,u\models \chi$, for all $\chi \in \sub(\varphi_1) \cup \sub(\varphi_2)$ and all $u\in W_{l}$, which is straightforward. For the second claim, we define $\bis\subseteq W_{1}\times W_{2}$ by setting 
	$((\type, \T_1, \T_2),(\type', \T_1', \T_2'))\in \bis$ if $\T_1=\T_{1}'$ and $\T_{2}=\T_{2}'$, for all $(\type, \T_1, \T_2)\in W_{1}'$ and 
	$(\type', \T_1', \T_2')\in W_{2}'$. We show that $\bis$ is a $\sigma$-bisimulation between $\Mmf_{1}'$ and $\Mmf_{2}'$. Assume $(\type,\T_{1},\T_{2})\in W_{1}$ and $(\type',\T_{1},\T_{2})\in W_{2}$ are given. Observe that all types in $\T_1\cup \T_2$ contain the same propositional variables and nominals in $\sigma$, since they are types realised in $\sigma$-bisimilar worlds in $\Mmf_{1}$ and $\Mmf_{2}$, respectively. Hence $(\type, \T_1, \T_2)$ and $(\type', \T_1, \T_2)$ satisfy the same propositional variables and nominals from $\sigma$ in $\Mmf_{1}'$ and $\Mmf_{2}'$, respectively. Now assume that $\B=(\type, \T_1, \T_2)$, $\B^{\ast}=(\type^{\ast},\T_{1}^{\ast},\T_{2}^{\ast})$, $\B'=(\type',\T_{1},\T_{2})$, and $(\B,\B^{\ast})\in R_{1}'$. By definition, there are $u,v\in W_{1}$ with $uR_{1}v$, $\B=\B_{\Mmf_{1}}(u)$, and $\B^{\ast}= \B_{\Mmf_{1}}(v)$. Also by definition, there exists $u'\in W_{2}$ with $\type_{\Mmf_{2}}(u')=\type'$ and $\Mmf_{1},u\sim_{\MLn(\sigma)} \Mmf_{2},u'$. By the definition of $\sigma$-bisimulations, there exists $v'$ such that $u'R_{2}v'$ and $\Mmf_{1},v\sim_{\MLn(\sigma)} \Mmf_{2},v'$. Then we have $(\B^{\ast},\B_{\Mmf_{2}}(v'))\in \bis$ and $(\B',\B_{\Mmf_{2}}(v'))\in R_{2}'$, as required. The converse condition is shown in the same way.      
\end{proof}

It follows that $\Mmf'_l, \B_{\Mmf_{l}}(u) \models \varphi_l$ and $\Mmf'_1, \B_{\Mmf_{1}}(w_1) \sim_{\MLn(\sigma)} \Mmf'_2, \B_{\Mmf_{2}}(w_2)$, as required.

Note that rather than building models directly from types we now build $\sigma$-bisimilar models from sets $\T_{\Mmf_{k}}(u)$ of types satisfied in $\sigma$-bisimilar nodes. The constructed models are of double-exponential size in $|\varphi_1| +|\varphi_2|$.
\end{proof}

The obtained double-exponential FBMP gives an obvious \textsc{coN2ExpTime} algorithm deciding the IEP for $\Kn$. This algorithm is not optimal. In fact, by unfolding the bisimilar models into bisimilar forest-shaped models (except for worlds denoted by nominals) and cutting them off at depth coinciding with the maximal modal depth of $\varphi$ and $\psi$, one can establish the exponential 
FBMP, which yields the following result:

\begin{theorem}[\cite{DBLP:journals/tocl/ArtaleJMOW23}]
The IEP for $\Kn$ is $\coNExpTime$-complete.
%
\end{theorem}


We illustrate the reason 
why the complexity of the IEP for $\Kn$ is higher than the complexity of its decision problem 
by sketching the idea of the lower bound construction. The proof is by reduction of the $\NExpTime$-complete $2^{n}\times 2^{n}$-tiling problem to the complement of the IEP for $\Kn$. Given any $n < \omega$ and a finite set $T$ of tile types, we define formulas $\varphi_{n,T}$ and $\psi_{n,T}$ of polynomial size in $|T|$ and $n$ such that $\varphi_{n,T} \to \psi_{n,T} \in \Kn$ and $T$ can tile a $2^{n}\times 2^{n}$-grid iff $\varphi_{n,T}$ and $\psi_{n,T}$ do not have an interpolant in $\Kn$. The common signature $\sigma_{n,T}$ of $\varphi_{n,T}$ and $\psi_{n,T}$ contains $O(n)$-many variables that  are used to speak about binary representations of cells in the $2^{n}\times 2^{n}$ grid; $\sigma_{n,T}$ also contains variables for the tile types in $T$. The formula $\varphi_{n,T}$ contains the conjunct $\Diamond^{2n}\i \wedge \Box^{2n}\i$, 
which is true in a pointed model $\Mmf,w$ iff the only world accessible from $w$ in $2n$ steps in $\Mmf$ is where nominal $\i$ holds. The formula $\psi_{n,T}$ contains standard conjuncts with $O(n)$-many non-$\sigma_{n,T}$ variables that generate a full binary tree of depth $2n$~\cite{DBLP:books/daglib/0030819,Blackburn_Rijke_Venema_2001} whose leaves $u_1,\dots, u_{2^{2n}}$ are regarded as the cells of the grid.  
Now, suppose that $\Mmf_1,w_1  \sim_{\MLn(\sigma_n)} \Mmf_2,w_2$, $\Mmf_1, w_1 \models \varphi_{n,T}$ and $\Mmf_2, w_2 \models \psi_{n,T}$. Then all the leaves $u_l$ of that full binary tree of depth $2n$ in $\Mmf_2$ are $\sigma_{n,T}$-bisimilar to the world where $\i$ is true in $\Mmf_1$, and so all trees with the roots $u_1,\dots, u_{2^{2n}}$ in $\Mmf_2$ are $\sigma_{n,T}$-bisimilar to each other. The other conjuncts of $\psi_{n,T}$ make sure that each tree rooted at $u_l$ describes the tile types of the grid cell represented by $u_l$ and of its neighbours in the grid; the $\sigma_{n,T}$-bisimilarity of these trees ensures that these local descriptions are consistent. 

\smallskip

\emph{Further Reading}. The IEP has been investigated in the context of learning 
description logic concepts from positive and negative data examples~\cite{DBLP:journals/ai/JungLPW22,DBLP:journals/tocl/ArtaleJMOW23}. Logics with nominals are particularly important in this application of interpolants as they are needed to encode constants used in description logic knowledge bases. We refer the reader to~\refchapter{chapter:kr}  for a detailed discussion. It is also of interest to note that no logic with the CIP that extends $\Kn$ and satisfies certain natural closure conditions is decidable~\cite{DBLP:journals/jsyml/Cate05}. Hence, the CIP cannot be restored easily by adding new connectives to $\Kn$ without sacrificing decidability. 



\subsection{Interpolant Existence in Fragments of First-Order Logic}

Locating expressive and yet decidable fragments of $\FO$ has been a popular research topic for more than a century~\cite{DBLP:books/sp/BorgerGG1997,Pratt23book}. It turns out, however, that, unlike full $\FO$, many of its prominent decidable fragments do not enjoy the CIP, and for these, the question arises as to whether the IEP is decidable, and, if so, whether its complexity is the same as that of entailment. We refer the reader to~\refchapter{chapter:predicate} and~\cite{DBLP:conf/fossacs/CateC24} for a detailed discussion of the landscape of fragments of $\FO$ with and without the CIP. 

The investigation of the IEP for $\FO$-fragments without the CIP has only just begun. Two prominent fragments of \FO{} that lack the CIP 
are the two-variable fragment $\FOT$ and the guarded fragment $\GF$. In $\FOT$, one restricts the number of individual variables in formulas to two. By re-using variables, one can show that $\ML$ (and even $\ML$ with nominals) are fragments of $\FOT$ under the standard translation~\cite{Blackburn_Rijke_Venema_2001} (see Section~\ref{Sec:2}).
Satisfiability of formulas in $\FOT$ is $\NExpTime$-complete, which can be shown by proving its exponential FMP~\cite{DBLP:books/sp/BorgerGG1997,Pratt23book}. By generalising the technique discussed in the previous section for $\Kn$, one can show that $\FOT$ has the double-exponential FBMP, for an appropriate notion of bisimulation characterising indistinguishability of models in $\FOT$. In contrast to the exponential FBMP of $\Kn$ and the exponential FMP of $\FOT$, this double exponential bound is optimal.

\begin{theorem}[\cite{DBLP:conf/lics/JungW21}]\label{thm:fO2 dec}
The IEP for $\FOT$ is in $\textsc{coN2ExpTime}$ and $2\ExpTime$-hard. 
\end{theorem}

The investigation of decidable extensions of $\FOT$ has been an important research topic, with results ranging from $\coNExpTime$-completeness of $\FOT$ with \emph{counting quantifiers} $\exists^{\ge k}x$, denoted $\CT$, to the decidability of $\FOT$ extended with one (built-in) transitive relation and undecidability of $\FOT$ with two transitive relations~\cite{Pratt23book}. At least one such extension provides an example with decidable entailment and undecidable IEP, namely $\FOT$ extended with two (built-in) equivalence relations whose  entailment is $\coTwoNExpTime$-complete.

\begin{theorem}[\cite{DBLP:conf/dlog/WolterZ24}]
The IEP for $\FOT$ with two equivalence relations is undecidable.
\end{theorem}

The decidability of the IEP for other extensions of $\FOT$ such as $\CT$ and $\FOT$ with one  equivalence relation remains open.

\paragraph*{Intermezzo} $\FOT$ is closely related to prominent decidable fragments of first-order modal logics with constant first-order domains, another family of logics lacking the CIP~\cite{DBLP:journals/jsyml/Fine79}. In fact, the one-variable fragment $\QSF$ of first-order modal logic $\SFi$ can be regarded as a fragment of $\FOT$. Exactly the same complexity bounds as in Theorem~\ref{thm:fO2 dec} for $\FOT$ are known for this logic, but with subtly different proofs~\cite{DBLP:conf/kr/KuruczWZ23}. In particular, the lower bound proof is more involved as it cannot use equality. Actually, this proof can be used to strengthen Theorem~\ref{thm:fO2 dec} since \TwoExpTime-hardness of the IEP for $\FOT$ even without equality can be shown as a corollary. The IEP of the one-variable fragment of first-order modal logic \K{} is also decidable, but in this case no elementary algorithm is known~\cite{DBLP:conf/kr/KuruczWZ23}.

\medskip

The guarded fragment $\GF$ of $\FO$ restricts quantification to guarded quantifiers of the form $\exists \avec{x}\, (R(\avec{x},\avec{y}) \wedge \varphi(\avec{x},\avec{y}))$. It also generalises modal logic and is decidable in $\ExpTime$ if the arity of predicate symbols is fixed and 2$\ExpTime$-complete otherwise. While its two-variable fragment, $\FOT\cap\GF$ enjoys the CIP, GF itself does not. Regarding the IEP, it is open whether it has the FBMP for the appropriate notion of bisimulation called \emph{guarded bisimulation}, but one can use instead models of bounded treewidth constructed from sets of types satisfied in guarded bisimilar tuples to show the following:

\begin{theorem}[\cite{DBLP:conf/lics/JungW21}]
The IEP for $\GF$ is $\TwoExpTime$-complete if the arity of predicates is bounded and $\textsc{3ExpTime}$-complete otherwise.
\end{theorem} 

\emph{Further Reading}. We close our discussion of the IEP with a few comments on Horn logics. While the Horn fragment of \FO{} enjoys the CIP~\cite[Exercise 6.2.6]{modeltheory}, Horn fragments of modal and description logics typically do not enjoy the CIP. The IEP for these logics is investigated in~\cite{DBLP:journals/corr/abs-2202-07186}, using a model-theoretic approach. Results on the exact reformulation of conjunctive queries modulo tuple-generating dependencies can also be formulated as results about the IEP in rule-based languages lacking the CIP~\cite{DBLP:series/synthesis/2016Benedikt,DBLP:conf/ijcai/BenediktKMT17}.


\section{Interpolation, Separation and Definability in Weaker Languages}\label{Sec:4}

In this section, we discuss a generalisation of Craig interpolation, which restricts the \emph{logical constructs} of the language for interpolants between formulas rather than the \emph{vocabulary} of non-logical symbols. All of our languages are assumed to be sublanguages of \emph{monadic second-order logic} $\MSO$ that extends $\FO$ with quantifiers of the form $\exists P$, where $P$ is any unary predicate symbol (thus, $\exists x$ quantifies over the elements of a domain in question, while $\exists P$ over the subsets of that domain). To simplify presentation, we also assume that our languages are closed under negation and that their signature comprises one binary predicate symbol, $R$, and arbitrarily-many unary ones. Unless indicated otherwise, we only consider formulas $\varphi(x)$ with at most one free variable $x$. 
Modal formulas are often regarded as their standard translations into $\FO$ or $\MSO$.  

For any such languages $L \supseteq L'$, we are interested in the following decision problem: 

\begin{description}
\item[\emph{$L/L'$-interpolation}:] given two $L$-formulas $\varphi$ and $\psi$, decide whether there exists an $L'$-formula $\chi$---called an $L'$-\emph{interpolant} for $\varphi$ and $\psi$---such that $\varphi \models \chi$ and $\chi \models \psi$.
\end{description}

\noindent
Note that, in this definition, we do not require the interpolant $\chi$ to be constructed from the shared non-logical symbols of $\varphi$ and $\psi$. In various areas of computer science, the interpolation problem is often formulated as a separation problem, which is sometimes more convenient to work with (see, e.g., Sections~\ref{sec:formal-languages} and~\ref{sec:LTL} below):

\begin{description}
\item[\emph{$\L/\LS$-separation}:] given $\L$-formulas $\varphi$ and $\psi$, decide whether there is an $\LS$-formula $\chi$--- called an $\LS$-\emph{separator} for $\varphi$ and $\psi$---such that $\varphi \models \chi$ and $\chi \models \neg \psi$.
\end{description}

Clearly, these two problems are equivalent to each other because $L$ is closed under (classical) negation.

\begin{example}
Consider the $\FO$-formulas 
\[	
\varphi(x) = \exists^{=1}y \, R(x,y), \quad \psi(x) = \exists^{=1}y\, \big(R(x,y) \wedge A(y)\big) \wedge \exists^{=1}y\, \big(R(x,y)\wedge \neg A(y)\big) 
\]
(where $\exists^{=1}y \, \alpha(y)$ says that there exists exactly one $y$ for which $\alpha(y)$ holds trues). Clearly, $\varphi \models \neg \psi$ and the $\ML$-formula $\chi = \Box A \vee \Box \neg A$ separates $\varphi(x)$ and $\psi(x)$. 
For $\psi'(x)= \exists^{=2}y\, R(x,y)$, we also have $\varphi\models \neg\psi'$, but $\varphi(x)$ and $\psi'(x)$ are not $\ML$-separable as follows from Theorem~\ref{thm:critindistsep} below. \lipicsEnd
\end{example}

$\L/\LS$-interpolation and separation are in turn generalisations of the more familiar 

\begin{description}
\item[\emph{$\L/\LS$-definability}:] given an $\L$-formula $\varphi$, decide whether there exists an $\LS$-formula $\chi$ such that $\varphi \models \chi$ and $\chi \models \varphi$.
\end{description}

\begin{example} The $\FO$-formulas  
\[
\varphi_{n}(x) = \exists y \, \big[ R(x,y) \wedge \bigwedge_{1\leq i \leq n} \big(A_{i}(x) \leftrightarrow A_{i}(y) \big) \big], \ \ n < \omega,
\] 
are $\ML$-definable. Indeed, let $T$ be the set of all conjunctions $B_{1}\wedge \dots \wedge B_{n}$, where $B_{i}\in \{A_{i},\neg A_{i}\}$, for $i=1,\dots,n$. Then each $\varphi_{n}$ is equivalent to $\chi_n = \bigvee_{t\in T}(t\wedge \Diamond t)$. It might be of interest to note that  these $\chi_n$ are of exponential-size in $n$ and there are no $\ML$-formulas $\chi_n$ that are equivalent to $\varphi_n$ and have polynomial-size $|\sub(\chi_n)|$~\cite{new25}. \lipicsEnd
\end{example} 

%

Below, we discuss definability and separation in fragments of \FO{}, including various types of modal logic, and in logics with second-order quantification or recursion such as $\MSO$ and the modal $\mu$-calculus $\muML$. While model-theoretic methods based on bisimulations are prominent for fragments of \FO{}, they have to be augmented by automata theoretic methods when dealing with second-order quantifiers. We also introduce an elegant `good model method' that has been used to study definability and construct defining formulas whenever they exist.  


\subsection{Definability and Separation in Fragments of First-Order Logic}
The next theorem is an analogue of Theorem~\ref{thm:critindistmodal}, characterising the existence of $\LS$-separators in terms of $\LS$-indistinguishable structures.
\begin{theorem}\label{thm:critindistsep}
For any $\LS \subseteq \L \subseteq \FO$ and any $\L$-formulas $\varphi(x)$ and $\psi(x)$, the following conditions are equivalent\textup{:}
\begin{itemize}
	\item there does not exist an $\LS$-separator for $\varphi(x)$ and $\psi(x)$\textup{;}
	
	\item there exist structures $\Amf_{i}$ and $a_{i} \in \dom(\Amf_i)$, $i=1,2$, such that $\Amf_{1}\models \varphi(a_{1})$, $\Amf_{2}\models\neg\psi(a_{2})$ and $\Amf_{1},a_{1}\equiv_{\LS(\sigma)} \Amf_{2},a_{2}$, where $\sigma = \sig(\varphi) \cup \sig(\psi)$.
\end{itemize}
%
\end{theorem}

Here, $\Amf_{1},a_{1}\equiv_{\LS(\sigma)} \Amf_{2},a_{2}$ means that $\Amf_{1} \models \gamma(a_{1})$ iff $\Amf_{2} \models \gamma(a_{2})$, for all $\LS(\sigma)$-formulas $\gamma$. This characterisation translates to the following criterion of $\L/\LS$-definability:

\begin{theorem}\label{thm:critindistdef}
Let $\LS \subseteq \L \subseteq \FO$. Then an $L$-formula $\varphi$ is logically equivalent to some $\LS$-formula iff $\varphi$ is preserved under $\equiv_{\LS(\sigma)}$, where $\sigma = \sig(\varphi)$.
\end{theorem}

As before, one can make these criteria operational by replacing $\equiv_{\LS(\sigma)}$ with a model-theoretic characterisation via suitable bisimulations. Below, we focus on definability and separation in $\ML$, and give pointers to the literature for more expressive languages. Using the fact that modal equivalence is characterised by bisimulations on $\omega$-saturated structures, we obtain van Benthem's theorem: 


\begin{theorem}[\cite{Benthem83}]\label{thm:vanBenthem}
For any $\FO$-formula $\varphi(x)$ the following conditions are equivalent\textup{:}
\begin{itemize}
	\item $\varphi(x)$ is equivalent to $($the standard translation of$)$ an $\ML$-formula\textup{;}
	
	\item $\varphi(x)$ is invariant under $\sigma$-bisimulations, 
	for $\sigma = \sig(\varphi)$, which means that, for any pointed models $\mathfrak M_i,u_i$, $i =1,2$, whenever $\mathfrak{M}_1,u_1 \sim_{\ML(\sigma)} \mathfrak{M}_2,u_2$, then $\mathfrak{M}_1 \models \varphi(u_1)$ iff $\mathfrak{M}_2 \models \varphi(u_2)$. 
\end{itemize}
\end{theorem}

We aim to apply Theorem~\ref{thm:vanBenthem} to the problem of deciding $\frag/\ML$-definability for $\FO$-fragments $\frag$. If $\frag = \FO$, this problem is undecidable because \FO{} is undecidable. For decidable fragments of $\FO$,  one can try to use the approach of the previous section and construct finite or tree-shaped well-behaved bisimilar models refuting the second item above. Here, we first follow a different path and use a more informative version of van Benthem's theorem whose proof does not use compactness but a more combinatorial approach (that also works on finite models, see below). It employs the notion of $n$-bisimilar pointed models~\cite{goranko20075}, which characterises modal indistinguishability up to formulas of modal depth at most $n$, where the \emph{modal depth} of an $\ML$-formula is the maximum number of nested modal operators in it. We write $\Mmf_{1},w_{1} \sim_{\ML(\sigma)}^{n} \Mmf_{2},w_{2}$ if $\Mmf_{1},w_{1}$ and $\Mmf_{2},w_{2}$ are $\sigma$-$n$-bisimilar, and  $\Mmf_{1},w_{1}\equiv_{\ML(\sigma)}^{n} \Mmf_{2},w_{2}$ if $\Mmf_{1},w_{1}$ and $\Mmf_{2},w_{2}$ satisfy the same $\ML(\sigma)$-formulas of modal depth $\le n$.  

\begin{lemma}
For any $n\geq 0$, we have $\Mmf_{1},w_{1}\sim_{\ML(\sigma)}^{n} \Mmf_{2},w_{2}$ iff $\Mmf_{1},w_{1}\equiv_{\ML(\sigma)}^{n} \Mmf_{2},w_{2}$. 
\end{lemma}

Now a non-trivial proof uses $n$-bisimilarity to bound the depth of an equivalent modal formula in terms of the depth of the \FO-formula~\cite{DBLP:journals/apal/Otto04}. We remind the reader that the \emph{quantifier rank} of an \FO-formula is the maximum nesting depth of quantifiers in it.

\begin{theorem}[\cite{DBLP:journals/apal/Otto04}]\label{thm:ottodefbent}
Any \FO-formula $\varphi(x)$ of quantifier rank $q$ that is invariant under bisimulations is equivalent to an $\ML$-formula of modal depth at most $2^{q}-1$.
\end{theorem} 

As one can compute a finite list of all $\ML$-formulas (up to logical equivalence) in a 
finite signature and of bounded modal depth, we obtain:

\begin{theorem}\label{thm:definabilityFO}
$\frag/\ML$-definability is decidable, for any decidable $\FO$-fragment $\frag$ containing $\ML$.
\end{theorem}

Note that this global approach, which applies to all decidable fragments of $\FO$ containing $\ML$, gives no elementary upper complexity bound since the number of modal formulas of a fixed modal depth is non-elementary. Moreover, for full $\FO$, there is no elementary bound on the size of an $\ML$-formula that is equivalent to a given $\FO$-formula~\cite{DBLP:journals/logcom/HellaV19}. We now discuss another technique that can sometimes be used to investigate $\L/\LS$-definability and even construct small defining formulas. 

\paragraph*{\emph{Good Model Method}} 
Let $\L\supseteq \LS$ be a pair of first-order or modal languages. A class $\mathcal{G}$ of pointed models is called \emph{good for $\L/\LS$-definability} if, for any $\L$-formula $\varphi$, there is a \emph{simplifying $\LS$-formula $\varphi^{s}$} such that 

\begin{description}
\item[($\LS$-transfer)] for every pointed model $\Mmf,w$ using symbols in a signature $\sigma$, there exists a pointed model $\Mmf^{s},w^{s}$ in $\mathcal{G}$ such that $\mathfrak{M},w \equiv_{\LS(\sigma)} \mathfrak{M}^{s},w^{s}$; 

\item[($\L/\LS$-indistinguishability)] for every pointed model $\Mmf,w$ in $\mathcal{G}$, we have $\mathfrak{M},w\models \varphi$ iff $\Mmf,w\models \varphi^{s}$.
\end{description}

\noindent
The following easily proved lemma shows that the good model method provides explicit $\LS$-definitions whenever they exist and reduces $\L/\LS$-definability to $\L$-validity. 

\begin{lemma}\label{lem:goodmodels}
If $\mathcal{G}$ is good for $\L/\LS$-definability, 
then an $\L$-formula $\varphi$ is $\LS$-definable iff $\models \varphi \leftrightarrow \varphi^{s}$.
\end{lemma}

We first illustrate the method for $\GML/\ML$-definability, which asks  whether counting modalities $\Diamond^{\ge k}$ in a given $\GML$-formula are really needed or can be replaced by standard modalities. The answer is surprisingly simple. For any Kripke model $\mathfrak{M}=(W,R,\V)$, let $\mathfrak{M}^{\omega}=(W^{\omega},R^{\omega},\V^{\omega})$ be the model obtained from $\Mmf$ by replacing every world $w$ in it by $\omega$-many copies $(w,i)$, $i < \omega$, that is, by taking
\[
W^\omega = \bigcup_{w\in W}\bar{w}, \quad R^{\omega} = \bigcup_{(w,v)\in R}\bar{w}\times \bar{v}, \quad 
\V^{\omega}(p) =  \displaystyle\bigcup_{w\in \V^{\omega}(p)} \bar{w},\quad \text{where } \bar{w} = \{(w,i)\mid 0\leq i < \omega\}.
\]
It is not difficult to show that the class of all pointed Kripke models of the form $\mathfrak{M}^{\omega},(w,0)$ is good for $\GML/\ML$-definability. For ($\ML$-transfer),
observe that $\Mmf,w$ and $\Mmf^{\omega},(w,0)$ satisfy the same $\ML$-formulas.
For ($\GML/\ML$-indistinguishability), we obtain from any $\GML$-formula $\varphi$ a simplifying formula $\varphi^{f}$ by replacing any $\Diamond^{\geq k}$ in $\varphi$ with $\Diamond$. We call $\varphi^{f}$ the \emph{flattening} of $\varphi$. It not difficult to show that $\Mmf^{\omega},(w,0)\models \varphi$ iff $\Mmf^{\omega},(w,0)\models \varphi^{f}$ since any world has infinitely many copies. Now, to decide $\ML$-definability of $\varphi\in \ML$, it is enough to check whether the $\GML$-equivalence $\varphi \leftrightarrow \varphi^{f}$ is true in all models, which can be done in $\PSpace$~\cite{Pratt23book}. Thus, we obtain:

\begin{theorem}[\cite{new25}]
$\GML/\ML$-definability is \PSpace-complete.
\end{theorem}

The good model method works well for various pairs $\L/\LS$, in which $\L$ is a logic with counting and $\LS$ without counting. For instance, it can be used to show that $\CT/\ML$-definability is co$\NExpTime$-complete \cite{new25}, where $\CT$ is the extension of $\FOT$ with counting quantifiers $\exists^{\ge k} x$~\cite{Pratt23book}. 
%

The good model methods can also be used to show Theorem~\ref{thm:ottodefbent}. The proof is in two steps by first analysing $\FO/\MLu$-definability and then $\MLu/\ML$-definability, where $\MLu$ is the extension of $\ML$ with the universal modality, denoted $\Diamond_{u}$ and interpreted by setting $\Mmf,w\models \Diamond_{u}\varphi$ if there exists $v\in W$ with $\Mmf,v\models \psi$. To obtain good models for $\FO/\MLu$-definability, one additionally unfolds $\Mmf^{\omega}$ into a tree-shaped model 
known as the \emph{$\omega$-expansion} $\Mmf^{\uparrow\omega}$ of $\Mmf$~\cite{DBLP:journals/igpl/AndrekaBN95,DBLP:conf/concur/JaninW96}. Given $\varphi(x)$ in $\FO$,
the simplifying formula $\varphi^{s}$ is the conjunction of all $\MLu$-consequences of $\varphi(x)$ of modal depth bounded by $2^{q}$ with $q$ the quantifier rank of $\varphi$, see ~\cite{Hollenberg}. Good models for $\MLu/\ML$-definability are obtained by taking for any pointed $\Mmf,w$ its disjoint union with some model $\Mmf$ satisfying every satisfiable $\ML$-formula (for instance the canonical model, see above). A simplifying formula is constructed from $\varphi$ in $\MLu$ by replacing any occurrence of $\Diamond_{u}\psi$ by $\top$ if $\psi$ is satisfiable and by $\bot$ otherwise, starting with innermost occurrences.

While definability has been studied in depth for many pairs $\L/\LS$ of logics, the work on $\L/\LS$-separation has only just begun. Computationally, separation is often much harder than definability. For example, using the fact that $\ML$-indistinguishability in Theorem~\ref{thm:critindistsep} can be replaced by $\ML$-bisimilarity, one can show that $\CT/\ML$-separation is undecidable by a reduction of Minsky machines. It follows that Theorem~\ref{thm:definabilityFO} does not generalise to separation. For fragments of $\CT$, however, a few decidability results are known. 
For instance, by proving the FBMP for $\GML$-formulas and $\ML$-bisimulations using 
an approach similar to the proof of Proposition~\ref{prop:FBMP}, one can
show that $\GML/\ML$-separation is $\coNExpTime$-complete---more complex than  $\GML/\ML$-definability but still decidable.

\begin{theorem}[\cite{new25}]\label{thm:GMLML}
\subtheorem{} $\CT/\ML$-separation is undecidable. 

\subtheorem{} $\GML/\ML$-separation is \coNExpTime-complete.
\end{theorem} 
\emph{Further Reading}. Theorems~\ref{thm:ottodefbent} and~\ref{thm:definabilityFO} have been generalised (sometimes with different bounds on the modal depth of the equivalent formula) to extensions such as graded modal logic~\cite{DBLP:journals/corr/abs-1910-00039}, modal logics with inverse or the universal modality~\cite{DBLP:journals/apal/Otto04}, and 
modal logic with nominals~\cite{10.1093/logcom/exaf025}. First results on the construction and size of separators have recently been obtained in~\cite{jung2025computation}. For instance, for $\GML/\ML$-separation in Theorem~\ref{thm:GMLML}, there always exist separators of elementary size if separators exist at all. It remains an open problem to determine tight bounds. In description logic, definability and separation have been studied at the level of TBoxes, which corresponds to the global consequence relation in modal logic: `if $\varphi$ is true in all worlds of a model, then $\psi$ is true in all worlds of that model'\!. Using model-theoretic methods such as bisimulations, some decidability and complexity results are obtained  in~\cite{LutEtAl11,ArtEtAl21}. A result to highlight is the undecidability of $\mathcal{ALCO}^{u}/\mathcal{ELO}$-separation (formulated in~\cite{ArtEtAl21} in terms of definability), which shows that separation can be undecidable for rather weak languages.



\subsection{Definability and Separation for Logics with Second-Order Quantifiers and Fixed-Point Operators}

Next, we consider two prominent logics, $\MSO$ and $\muML$, that add expressive power to $\FO$ and $\ML$ in the form of quantifiers over sets and fixed-point operators, respectively, and are used in formal verification and many other areas of computer science.  
Logics with second-order quantifiers or fixed points are typically not compact because of which  some of the model-theoretic methods discussed 
in the previous section do not work, at least when applied to standard relational structures. In this case, automata-theoretic approaches, possibly based on finitary approximations of bisimulations, are often more appropriate. 

In this section, we mainly focus on two questions:
\begin{itemize}
\item Can the analysis of definability and separation for $\FO/\ML$ be lifted to $\MSO/\muML$? 

\item How hard is it to decide definability and separation for formulas with fixed-point operators by formulas without those operators?
\end{itemize}
In many cases, the latter question is equivalent to asking whether the use of recursion in a given formula is spurious and can be eliminated. We begin by reminding the reader of the basic $\mu$-calculus.

Given a set $W$, a function $f\colon 2^{W} \rightarrow 2^{W}$ is called \emph{monotone} if $X\subseteq Y$ implies $f(X)\subseteq f(Y)$. The \emph{least fixed point} of $f$ is the smallest (w.r.t.\ $\subseteq$) set $X\subseteq W$ such that $f(X)=X$. If $f$ is monotone, the least fixed point of $f$ always exists and can be computed by transfinite induction as follows:
\[
f^{0}(X)=X, \quad f^{\alpha+1}(X)=f(f^{\alpha}(X)), \text{ for any ordinal $\alpha$,} \quad f^{\alpha}(X)=\bigcup_{\beta<\alpha}f^{\beta}(X),
\text{ for a limit ordinal $\alpha$.}
\]
The least fixed point of $f$ is then given by $f^{\alpha}(\emptyset)$, for the least ordinal $\alpha$ with $f^{\alpha}(\emptyset)=f^{\alpha+1}(\emptyset)$; this $\alpha$ is called the \emph{closure ordinal} of $f$. 

The \emph{modal $\mu$-calculus} $\muML$ is the extension of basic modal logic $\ML$ with least fixed-point operators. More precisely, \emph{$\mu\ML$-formulas} are defined  by the grammar 
\[
\varphi \ \ := \ \ \top \ \mid \ p \ \mid \ \varphi \wedge \varphi' \  \mid \ \neg \varphi \ \mid \ \Diamond \varphi \ \mid \ P \ \mid \ \mu P.\varphi,  
\]
where $p$ ranges over propositional variables and $P$ over unary predicate symbols (second-order variables). For formulas of the form $\mu P.\varphi$, it is required that every free occurrence of $P$ in $\varphi$ is \emph{positive} (i.e., in the scope of an even number of $\neg$). A \emph{$\mu\ML$-sentence} is any $\muML$-formula $\varphi$ whose  every second-order variable $P$ is in the scope of some $\mu P$.
Given a $\muML$-formula $\varphi$, a Kripke model $\mathfrak M = (W,R,\V)$, $u \in W$, and a function $\mathfrak V$ mapping second-order variables $P$ to $2^{W}$, the \emph{truth-relation} $\Mmf,u\models_{\mathfrak V}\varphi$ is defined inductively by extending the respective definition for $\ML$ with two items:
\begin{itemize}
\item $\Mmf,u\models_{\mathfrak V} P$ iff $u \in \mathfrak V(P)$;

\item $\Mmf,u\models_{\mathfrak V} \mu P. \varphi$ iff $u$ is in the least fixed point of the function $f_{\varphi} \colon X \mapsto \{ v\in W \mid \Mmf,v\models_{\mathfrak V[P\mapsto X]} \varphi\}$, where $\mathfrak V[P\mapsto X]$ maps $P$ to $X$ and coincides with $\mathfrak V$ on all other second-order variables.
\end{itemize}
(Note that $f_{\varphi}$ is monotone because $P$ occurs only positively in $\varphi$.) For further details and discussions, we refer the reader to~\cite{DBLP:books/el/07/BradfieldS07}. Here we only show some standard examples of what can be said in $\muML$.


\begin{example}
Using the definitions of truth and least fixed point, one can readily see that 
\begin{itemize}
\item $\Mmf,u\models \mu P. (\Diamond q \lor \Diamond P)$ iff there is an $R$-path $u R u_1 R \dots R u_n$ with $n \ge 1$ in $\mathfrak M$ such that $\Mmf,u_n\models q$;

\item $\Mmf,u\models \mu P. \big(\Diamond q \lor (r \land \Diamond P)\big)$ iff there is an $R$-path as above, in which $\Mmf,u_i\models r$ for all $i$, $1 \le i < n$;

\item $\Mmf,u\models \mu P. \Box P$ iff $\mathfrak M$ does not contain an infinite path of the form $u R u_1 R \dots R u_n R \dots$.
\end{itemize}
In fact, $\muML$ is powerful enough to express the operators of temporal logics $\LTL$ (see Section~\ref{sec:LTL} below), $\CTL$, $\CTL^*$ and propositional dynamic logic $\PDL$ (in which case a multimodal $\muML$ is required). \lipicsEnd
\end{example}

In many ways, $\mu\ML$ has the same relationship to $\MSO$ as $\ML$ has to $\FO$. First, we observe that the standard translation $\varphi \mapsto \varphi_{x}^{\ast}$ of $\ML$ to $\FO$ defined in Section~\ref{Sec:2} can naturally be extended to a translation of $\mu\ML$ to $\MSO$ by taking 
\[
P_x^* = P(x), \quad (\mu P . \varphi)_x^* = \forall P \, \big[\forall \bar x\, \big (\varphi^*_{\bar x} \leftrightarrow P(\bar x) \big) \to P(x) \big], \text{ where $\bar x = y$, $\bar y = x$}.
\]
Second, $\mu\ML$-sentences are still invariant under $\ML$-bisimulations and, in fact, we have the following extension of Theorems~\ref{thm:vanBenthem} and~\ref{thm:definabilityFO}:

\begin{theorem}[Janin-Walukiewicz Theorem \cite{DBLP:conf/concur/JaninW96}]\label{thm:JaninWal}
\subtheorem{} An $\MSO$-formula $\varphi(x)$ is $\mu\ML$-definable iff it is invariant under $\ML$-bisimulations.

\subtheorem{} ${\sf M}/\mu\ML$-definability is decidable, for any decidable $\MSO$-fragment ${\sf M}$ containing $\mu\ML$.
\end{theorem}

The proof is by the good model method. The good models are the tree-shaped $\omega$-expansions $\Mmf^{\uparrow\omega}$ introduced above. However, it is harder now to construct appropriate simplifying formulas. They are obtained by translating the input $\MSO$-formula $\varphi(x)$ to an automaton $\mathcal{A}$ accepting the tree-shaped models of $\varphi$. The automaton $\mathcal{A}$ is then transformed to an automaton $\mathcal{A}^{s}$ accepting exactly the same models as a $\muML$-sentence $\varphi^{s}$ and exactly the same models of the form $\Mmf^{\uparrow\omega}$ as $\mathcal{A}$. Similarly to the flattening $\varphi^{f}$ defined above for $\GML/\ML$, the transformation from $\mathcal{A}$ to $\mathcal{A}^{s}$ forgets the quantitative information in $\mathcal{A}$. The construction is effective, and so, by Lemma~\ref{lem:goodmodels}, one can decide $\mu\ML$-definability of $\varphi$ for decidable fragments of \MSO{} by checking the equivalence $\models \varphi \leftrightarrow \varphi^{s}$. 

\smallskip

\emph{Further Reading}. Different methods are needed to show Theorems~\ref{thm:vanBenthem} and~\ref{thm:JaninWal} on finite models. For Theorem~\ref{thm:vanBenthem}, this was first achieved by Rosen~\cite{DBLP:journals/jolli/Rosen97}. Otto's proof of Theorem~\ref{thm:ottodefbent} can also be used for this purpose. Theorem~\ref{thm:JaninWal} has only very recently been shown on finite models~\cite{DBLP:journals/corr/abs-2407-12677}, see also~\cite{DBLP:journals/tcs/BlumensathW20,DBLP:conf/lics/PfluegerMK24}. 

\smallskip

We next consider the definability and separation of formulas with least fixed-point operators by formulas without such operators. A basic variant of this problem is the \emph{boundedness problem} that has been investigated for a few decades in many different contexts. Suppose $\varphi(P,x)$ is an $\FO$-formula with a unary predicate $P$ that occurs only positively in $\varphi$. Given a structure $\mathfrak A$, define a monotone function $f_{\varphi} \colon 2^{\dom(\Amf)} \rightarrow 2^{\dom(\Amf)}$ by taking $f_{\varphi}(X)=\{ a\in \dom(\Amf) \mid \Amf\models \varphi(X,a)\}$. 
Similarly to $\mu\ML$ we can introduce the fixed point operator $\mu P.\varphi(P,x)$ with the following truth-condition: $\Amf\models \mu P.\varphi(a)$ iff $a$ belongs to the least fixed point of $f_{\varphi}$ on $\Amf$.

\begin{example}
For $\varphi(P,x) = A(x) \lor \exists y\, \big(R(x,y) \land P(y) \big)$, we have $\Amf\models \mu P.\varphi(a)$ iff $\Amf$ contains an $R$-path from $a$ to some $b \in \dom(\Amf)$ with $\Amf\models A(b)$. Thus, the formula $\mu P.\varphi(x)$ is equivalent to the monadic datalog query $(\Pi_\varphi,P(x))$, in which the program $\Pi_\varphi$ consists of two rules: $P(x) \leftarrow A(x)$ and $P(x) \leftarrow R(x,y) \land P(y)$. \lipicsEnd
\end{example}

A formula $\varphi(P,x)$ is called \emph{bounded} if there is some $n < \omega$ such that, for every structure $\Amf$, the closure ordinal of $f_{\varphi}$ applied to $\Amf$ does not exceed $n$. It is readily seen that if $\varphi(P,x)$ is bounded, then $\mu P.\varphi(P,x)$ is $\FO$-definable. The classical Barwise-Moschovakis theorem below states that the converse also holds for $\FO$:

\begin{theorem}[\cite{DBLP:journals/jsyml/BarwiseM78}]
For every $\FO$-formula $\varphi(P,x)$ with only positive occurrences of $P$, the formula $\mu P.\varphi$ is $\FO$-definable iff $\varphi(P,x)$ is bounded over the class of all structures. 
\end{theorem}

This notion of boundedness formalises a fundamental question in computer science: when is recursion really needed? 
For example, is it possible to transform (optimise) a given database query with recursion (say, a datalog query) to an equivalent SQL-query, i.e., essentially an $\FO$-formula? 

\begin{example}\label{nasty}
One can check that, for $\psi(P,x)$ depicted below, the closure ordinal of $f_\psi$ is 2: 

\vspace*{2mm}


$\psi(P,x) = A(x) \lor \exists w,v,u,z,y\, \big($ 

\vspace*{-6.7mm}
\hspace*{4.2cm}\begin{tikzpicture}[>=latex,line width=0.8pt,rounded corners, scale = 0.9]
\node[point,scale = 0.5,label=below:{\scriptsize $w$}] (0) at (-1.5,0) {};
\node[point,scale = 0.5,label=below:{\scriptsize $v$}] (1) at (0,0) {};
\node[point,scale = 0.5,label=above:{\scriptsize $P$},label=below:{\scriptsize $u$}] (m) at (1.5,0) {};
\node[point,scale = 0.5,label=below:{\scriptsize $z$}] (2) at (3,0) {};
\node[point,scale = 0.5,label=above:{\scriptsize $A$},label=below:{\scriptsize $y$}] (3) at (4.5,0) {};
\node[point,scale = 0.5,label=below:{\scriptsize $x$}] (4) at (6,0) {};
\draw[<-,right] (0) to node[below] {\scriptsize $R$}  (1);
\draw[<-,right] (1) to node[below] {\scriptsize $R$}  (m);
\draw[<-,right] (m) to node[below] {\scriptsize $R$} (2);
\draw[->,right] (2) to node[below] {\scriptsize $R$} (3);
\draw[->,right] (3) to node[below] {\scriptsize $R$} (4);
\end{tikzpicture} 

\vspace*{-8.5mm}
\hspace*{11.34cm}$\big)$.

\vspace*{3mm}
\noindent
The formula $\mu P. \psi(P,x)$ corresponds to the datalog query $(\Pi_\psi,P(x))$, where  $\Pi_\psi$ contains the rules
\[
P(x) \leftarrow A(x), \quad P(x) \leftarrow R(v,w) \land R(u,v) \land P(u) \land R(z,u) \land R(z,y) \land A(y) \land R(y,x),
\]
and so this datalog query can be rewritten to an equivalent $\FO$-query. \lipicsEnd
\end{example}


The problem of deciding whether $\mu P.\varphi$ is $\FO$-definable, for  an input formula $\varphi(P,x)$ in a given language $\L$, has been investigated extensively, mostly in the equivalent boundedness formulation. For example, this problem is undecidable for $\FOT$~\cite{DBLP:conf/lics/KolaitisO98} but decidable for $\GF$~\cite{DBLP:conf/lics/BenediktCCB15}.
For a wide range of results on the computational complexity of deciding boundedness of datalog queries and ontology-mediated queries in various description logics the reader can consult, e.g.,~\cite{DBLP:conf/stoc/CosmadakisGKV88,DBLP:journals/tods/BienvenuCLW14,DBLP:journals/tocl/BenediktBGS20,DBLP:journals/ai/LutzS22} and further references therein. For instance, deciding boundedness of $\mu P.\varphi$ such as the one in Example~\ref{nasty} is $\TwoExpTime$-complete~\cite{DBLP:conf/lics/BenediktCCB15,DBLP:conf/pods/KikotKPZ21}. 



We next consider a more general case, where fixed-point operators are part of the language $L$ (and so can be nested), focusing on $\mu\ML/\ML$-definability and separation.


\begin{theorem}[\cite{DBLP:conf/stacs/Otto99,DBLP:conf/dlog/JungK24}]\label{thm:muml}
$\mu\ML/\ML$-definability and separation are both \ExpTime-complete.
\end{theorem}

To discuss the proof, given any signature $\sigma$, let $\mathfrak{M}_{1},w_{1} =_{\sigma}^{n}\mathfrak{M}_{2},w_{2}$ denote that $\mathfrak{M}_{1}$ and $\mathfrak{M}_{2}$ are tree-shaped with roots $w_{1}$ and $w_{2}$, respectively, and the $\sigma$-reducts of the submodels induced by the set of nodes in $\mathfrak{M}_{1}$ and $\mathfrak{M}_{2}$ reachable from $w_{1}$ and $w_{2}$ in $n$ steps are isomorphic. The next lemma gives the key steps of the proof:

\begin{lemma}
For any $\mu\ML$-sentences $\varphi_{1}$ and $\varphi_{2}$, $n\ge 0$, and finite signature $\sigma$, the following conditions are equivalent\textup{:}
\begin{enumerate}
\item[$\mathbf{(a)}$] there is no $\ML(\sigma)$-formula of modal depth $\leq n$ separating $\varphi_{1}$ and $\varphi_{2}$\textup{;}

\item[$\mathbf{(b)}$] there are tree-shaped $\mathfrak{M}_{1},w_{1}\models \varphi_{1}$ and $\mathfrak{M}_{2},w_{2}\models\neg\varphi_{2}$ such that $\mathfrak{M}_{1},w_{1}\equiv_{\ML(\sigma)}^{n}\mathfrak{M}_{2},w_{2}$\textup{;}

\item[$\mathbf{(c)}$] there are tree-shaped  $\mathfrak{M}_{1},w_{1}\models \varphi_{1}$ and $\mathfrak{M}_{2},w_{2}\models\neg\varphi_{2}$ such that $\mathfrak{M}_{1},w_{1}\sim_{\ML(\sigma)}^{n}\mathfrak{M}_{2},w_{2}$\textup{;}

\item[$\mathbf{(d)}$] there are tree-shaped  $\mathfrak{M}_{1},w_{1}\models \varphi_{1}$ and $\mathfrak{M}_{2},w_{2}\models\neg\varphi_{2}$ such that $\mathfrak{M}_{1},w_{1}=_{\sigma}^{n}\mathfrak{M}_{2},w_{2}$\textup{;}

\item[$\mathbf{(e)}$] there are tree-shaped  $\mathfrak{M}_{1},w_{1}\models \varphi_{1}$ and $\mathfrak{M}_{2},w_{2}\models\neg\varphi_{2}$ of outdegree $\leq |\varphi_{1}|+|\varphi_{2}|$ such that $\mathfrak{M}_{1},w_{1}=_{\sigma}^{n}\mathfrak{M}_{2},w_{2}$.
\end{enumerate}
\end{lemma}

Observe that, for $\sigma=\sig(\varphi)\cup \sig(\varphi)$, $\mathbf{(a)}$ holds for all $n>0$ iff there is no $\ML$-separator for $\varphi_{1},\varphi_{2}$. The main step is $\mathbf{(c)} \Rightarrow \mathbf{(d)}$, which is an amalgamation property that typically fails when dealing with logics that lack the CIP or when investigating $\L/\LS$-separation. The \ExpTime{} upper bound is derived from $\mathbf{(e)}$ using automata-theoretic techniques by showing that $\mathbf{(e)}$ holds for all $n>0$ iff $\mathbf{(e)}$ holds 
for some $n$ exponential in $|\varphi_{1}|+|\varphi_{2}|$. 
Interestingly, $\mathbf{(c)} \Rightarrow \mathbf{(d)}$ does not hold if the outdegree of models is bounded and $\geq 3$, which reflects that $\Alt_{n}$ does not have the CIP for $n\geq 3$. Also, $\mu\ML/\ML$-separation on $n$-ary trees, $n\geq 3$, is 2\ExpTime-complete, a result that is harder to show than the \ExpTime-bound above~\cite{Jeanstacs}.

\smallskip

\emph{Further Reading}. Many definability results discussed in this chapter have been (at least partially) extended to guarded fragments of $\FO$ with and without fixed points. Theorem~\ref{thm:vanBenthem} has been lifted to a characterisation of $\GF$ within $\FO$ using \emph{guarded bisimulations}~\cite{ANvB98}. It appears to be open whether an analogue of Theorem~\ref{thm:definabilityFO} holds for $\GF$, but rather general positive results shown by reduction to languages over trees are known~\cite{DBLP:journals/lmcs/BenediktBB19}.
For instance, it is \TwoExpTime{}-complete to decide whether a formula in the guarded negation fragment of $\FO$ is definable in $\GF$.  
The characterisation part of Theorem~\ref{thm:JaninWal} has been lifted, using guarded bisimulations, to a characterisation
of $\mu\GF$ within a fragment $\GSO$ of $\SO$ whose second-order quantification is restricted to guarded sets~\cite{DBLP:journals/tocl/GradelHO02}. Note that the characterisation does not hold if $\GSO$ is replaced by full $\SO$. Again, it appears to be open whether the decidability part of Theorem~\ref{thm:JaninWal} still holds, but general positive results are shown in~\cite{DBLP:journals/lmcs/BenediktBB19}.

%
%
%

\section{Definability and Separation in Formal Languages}\label{sec:formal-languages}

In this section, we consider the following decision problems 
for classes $\mathcal{C}'\subsetneqq\mathcal{C}$ of formal languages:

\begin{description}
\item[\emph{$\mathcal{C}/\mathcal{C}'$-definability}:] given a language $\lang\in\mathcal{C}$, decide whether $\lang \in \mathcal{C}'$.
\end{description}
\begin{description}
\item[\emph{$\mathcal{C}/\mathcal{C}'$-separation}:] 
given languages $\lang_1,\lang_2\in\mathcal{C}$, decide whether there is a language $\lang\in\mathcal{C}'$---called a $\mathcal{C}'$-\emph{separator} for $\lang_1$ and $\lang_2$---such that $\lang_1\subseteq\lang$ and $\lang_2\cap\lang=\emptyset$.
\end{description}
A significant amount of related work has been focusing on defining and separating regular languages.
The study of regular languages brings together combinatorial methods on words with automata-theoretic, algebraic and logical techniques. These were
originally developed for languages consisting of finite words, but have been generalised to languages over infinite words, trees, and graphs.

Here, we discuss in detail a particular instance of these problems: the decidability of defining and separating regular languages consisting of finite words by star-free (or, equivalently, \FO-definable)  languages.
We also discuss possible bounds on the size of a definition/separator when it exists.
Some topics in the wider context 
(such as
languages consisting of infinite words and various kinds of labelled trees; 
using some $\FO$-fragments as definitions/separators;
and defining/separating non-regular languages with regular ones)
are only briefly sketched.

We begin by recalling the basic notions from the theory of formal languages; see, e.g.,~\cite{DBLP:books/daglib/0016921,DBLP:books/daglib/0086373}. 
An \emph{alphabet} is any finite set $\abc$.
Given an alphabet $\abc$, we denote by $\abc^+$ the set of all nonempty finite words over $\abc$. A subset 
$\lang\subseteq\abc^+$ is called a \emph{language\/} \emph{over} $\abc$.
(To simplify presentation, we only consider languages without the empty word, but similar results hold if the empty word is also allowed.)
The class of \emph{regular languages} is the smallest class that
contains $\abc^+$, all singletons $\{a\}$, for $a \in\abc$,
and is closed under concatenation, finite union, and the +-version of the Kleene star (iteration). 
The class of \emph{star-free languages} is the smallest class that 
contains $\abc^+$, all singletons $\{a\}$, for $a \in\abc$,
and is closed under concatenation, finite union, and complementation.
(As the class of regular languages is closed under complementation, star-free languages
are regular.)


{\bf Automata-theoretic approach to regular languages.}
It is well known that a language is regular iff 
it can be specified by a deterministic finite automaton (DFA) or, equivalently, by a nondeterministic finite automaton (NFA). 
Given an NFA $\mathfrak A$, every word $w\in\abc^+$ determines a 
binary \emph{transition relation} $\Delta_w$ on the states of $\mathfrak A$:
it consists of all those state-pairs $(q,q')$, for which $\mathfrak A$ can reach $q'$ having started at $q$ and read $w$. If $\mathfrak A$ is a DFA, then each $\Delta_w$ is actually a state-to-state \emph{transition function} $\delta_w$.
We say that such a $\delta_w$ \emph{has a non-trivial cycle} if there exist a state $q$ and $n$,  $1<n<\omega$, such that
$\delta_w(q)\ne q$ and $\delta_{w^n}(q)=q$.
A DFA is called \emph{counter-free} if no $\delta_w$ has a non-trivial cycle, for any $w\in\abc^+$.

{\bf Algebraic approach to regular languages.}
A \emph{semigroup} is a structure $\sg=(S,\sgo)$, where $\sgo$ is an associative binary operation on $S$.
For example, the following are semigroups: $\abc^+$ with concatenation (also denoted by $\abc^+$),  
and the set of transition relations in an NFA
(transition functions in a DFA) $\mathfrak A$ with composition, called the \emph{transition semigroup of} $\mathfrak A$.
Given $s\in S$ and $n$ with $0<n<\omega$, we write $s^n$ for $\underbrace{s\sgo \dots \sgo s}_n$.
We call $s$  \emph{idempotent} if $s^2=s$.
It is easy to see that, for every finite semigroup $\sg$, there is $\ino{\sg}$, $0<\ino{\sg}\leq |S|!$, 
such that $s^{\ino{\sg}}$ is idempotent,
for all $s$ in $\sg$. We call 
$\sg$ \emph{aperiodic}  if $s^{\ino{\sg}}=s^{\ino{\sg}+1}$, for all $s\in S$.

We say that a language $\lang$ is \emph{recognised} by a semigroup $\sg$ if there exist a semigroup homomorphism 
$\alpha\colon\abc^+ \to\sg$ and some $F \subseteq S$ such that $\lang = \alpha^{-1}(F)$. 
It is easy to see that if $\lang$ is recognised by a finite semigroup, then $\lang$ is regular. 
On the other hand, any language $\lang$ that is specified by an NFA $\mathfrak A$ is always recognised by the transition semigroup of $\mathfrak A$. As the latter is finite (of size at most $2^{n^2}$, for the number $n$ of states in $\mathfrak A$), we obtain:

\begin{theorem}\label{thm:regalg} 
A language is regular iff it is recognised by a finite semigroup.
\end{theorem}  

In fact, for any regular language $\lang$, one can obtain a `canonical' finite semigroup recognising $\lang$, called the \emph{syntactic semigroup of} $\lang$: it is the transition semigroup  of the minimal DFA specifying $\lang$.

%
%


{\bf Logical approach to regular languages.}
Any word $w\in\abc^+$ can be regarded as an $\FO$-structure in the signature $\sigma_\abc$ consisting of a binary predicate symbol
$<$ (interpreted in $w$ as the strict linear order over its positions $\{0,\dots,|w|-1\}$) and unary predicate symbols $\ap$, for $a\in\abc$ (selecting the subset of  $\{0,\dots,|w|-1\}$ where the symbol in $w$ is $a$).
A language $\lang\subseteq\abc^+$ is said to be $\MSOo$-\emph{definable} if there is some
$\MSO$-sentence $\varphi$ (in which all first- and second-order variables are bound) in the signature $\sigma_\abc$ such that $\lang=\{w\in\Sigma^+\mid w\models\varphi\}$.
If such a $\varphi$ can be chosen to be an $\MSO$-sentence of the form $\exists P_1\dots\exists P_n\,\psi$ with all $P_i$ being unary predicate symbols and $\psi$ not having any second-order quantification, 
then $\lang$ is called $\exists\MSOo$-\emph{definable}; and if $\varphi$ can be an $\FO$-sentence, then $\lang$ is called $\FOo$-\emph{definable}.
The following statement is known as the B\"uchi--Elgot--Trakhtenbrot theorem (see also~\cite[Theorem III.1.1]{Straubing94} and \cite[Theorem 7.21]{DBLP:books/sp/Libkin04}):

\begin{theorem}[\cite{Buchi60,Elgot61,Trakh62}]\label{thm:buchi}
A language is regular iff it is $\MSOo$-definable iff it is $\exists\MSOo$-definable.
\end{theorem}  


{\bf Languages over infinite words.}
Given an alphabet $\abc$, we denote by $\abc^\omega$ the set of all infinite words, called $\omega$-\emph{words\/}, over $\abc$. A subset $\lang\subseteq\abc^\omega$ is called an $\omega$-\emph{language\/} \emph{over} $\abc$.
The notions of
\emph{regularity} and \emph{star-freeness} of $\omega$-languages can be defined by appropriately adjusting the closure conditions given above for finite-word languages. The automata-theoretic, algebraic and logical approaches can be adapted as well. 
Regular $\omega$-languages are those that are specified by \emph{nondeterministic B\"uchi automata} (\emph{NBAs})~\cite{Buchi60}.
Theorem~\ref{thm:regalg} also holds in this setting: 
regular $\omega$-languages coincide with those that are recognised by finite $\omega$-\emph{semigroups}  (an appropriate generalisation of finite semigroups)~\cite{DBLP:books/daglib/0016866}.
Also, any $\omega$-word $w\in\abc^\omega$ can be regarded as an $\FO$-structure in the signature $\sigma_\abc$ whose
underlying strict linear order is isomorphic to $(\omega,<)$,  and Theorem~\ref{thm:buchi} also holds:
An $\omega$-language over $\abc$ is regular iff it is definable over $(\omega,<)$ by an $\MSO$-sentence in the signature $\sigma_\abc$~\cite{Buchi60}.



\subsection{Definability}

The problem of deciding whether a regular language is star-free can be 
addressed via giving automata-theoretic, algebraic, or logical characterisations of star-free languages:

\begin{theorem}[\cite{McNaughton&Papert71,DBLP:journals/iandc/Schutzenberger65a,DBLP:journals/iandc/Stern85}]\label{thm:FO} Suppose $\lang$ is a regular language. Then the following are equivalent\textup{:}

\subtheorem{thm:FOsf} 
$\lang$ is star-free\textup{;}

\subtheorem{thm:FOaut} 
$\lang$ is specified by a counter-free DFA 
\textup{(}or, equivalently, the minimal DFA specifying $\lang$ is counter-free\textup{)}\textup{;}

\subtheorem{thm:FOalg} 
$\lang$ is recognised by a finite aperiodic semigroup
\textup{(}or, equivalently, the syntactic semigroup of $\lang$ is aperiodic\textup{)}\textup{;}

\subtheorem{thm:FOlog} 
$\lang$ is $\FOo$-definable.
\end{theorem}

Observe that, by Theorems~\ref{thm:buchi} and \ref{thm:FOlog}, 
the question whether a regular language is star-free corresponds to asking
whether an $\MSO$-sentence in the signature $\sigma_\abc$ is equivalent to 
some 
$\FO$-sentence over finite strict linear orders.
There are many other properties equivalent to the ones listed in Theorem~\ref{thm:FO}; see \cite{DBLP:conf/birthday/DiekertG08} for a survey.
(As we shall see in the next section, definability in linear temporal logic $\LTL$ over finite timelines is one of them.)
Note also that there are several other equivalent algebraic definitions of aperiodicity, leading to different proofs of the equivalence of the other properties and Theorem~\ref{thm:FOalg}.
See Examples~\ref{ex:yes} and \ref{ex:no} below for some languages that are $\FOo$-definable/undefinable.

Using the automata-theoretic and algebraic characterisations of $\FOo$-definability, one can show that it is decidable and establish a tight complexity bound:

\begin{corollary}[\cite{DBLP:journals/iandc/Stern85,DBLP:journals/tcs/ChoH91,DBLP:journals/actaC/Bernatsky97}]\label{co:defcom}
Deciding $\FOo$-definability of the language specified by a DFA $\mathfrak A$ is $\PSpace$-complete in the size of $\mathfrak A$.
\end{corollary}

{\bf Constructing a possible definition.}
Given a counter-free DFA $\mathfrak A$ specifying $\lang_{\mathfrak A}\subseteq\abc^+$,
Wilke~\cite{DBLP:conf/stacs/Wilke99} constructs an $\FO$-definition $\varphi$ of $\lang_{\mathfrak A}$ such that
the size of $\varphi$ is linear in $|\abc|$ and double-exponential in the size of $\mathfrak A$.

{\bf Languages over infinite words.}
$\FOo$-definable $\omega$-languages also have equivalent characterisations that are similar to those in Theorem~\ref{thm:FO}.
The equivalence of $\FOo$-definability and star-freeness is shown in~\cite{DBLP:journals/iandc/Ladner77,DBLP:journals/iandc/Thomas79}.
Given any NBA specifying an $\omega$-language $\lang$, one can obtain a `canonical' structure (the \emph{syntactic $\omega$-semigroup} of $\lang$)~\cite{DBLP:conf/icalp/Wilke91}, whose \emph{aperiodicity} (that is, a property similar to aperiodicity of semigroups) is equivalent to
$\lang$ being star-free~\cite{DBLP:conf/mfcs/Perrin84}.
The equivalence of $\FOo$-definability and being specified by deterministic counter-free B\"uchi automata with Rabin--Streett acceptance conditions is shown in~\cite{DBLP:journals/iandc/Thomas81}, based on earlier work in~\cite{DBLP:journals/iandc/McNaughton66}. 
The survey \cite{DBLP:conf/birthday/DiekertG08} also covers several other equivalent characterisations.

{\bf Languages over trees.}
The question of deciding whether a given language consisting of various versions of labelled trees is definable in first-order logic is a long-standing open problem in language theory~\cite{DBLP:conf/caap/Thomas84}.
There have been a lot of (unsuccessful) attempts to fully generalise Theorem~\ref{thm:FO} for tree-languages;
see, e.g., the surveys~\cite{DBLP:journals/corr/abs-2008-11635,DBLP:books/ems/21/Bojanczyk21}.



\subsection{Separation}


As the class $\mathcal{C}$ of regular languages is closed under complementation, $\mathcal{C}/\mathcal{C}'$-definability can be reduced to $\mathcal{C}/\mathcal{C}'$-separation, for any $\mathcal{C}' \subseteq \mathcal{C}$: 
definability of $\lang$ can be decided by testing separability of $\lang$ from its complement. 
Also, by the above, the problem of deciding whether two regular languages are separable by a star-free language 
has equivalent automata-theoretic, algebraic and logical formulations:
%
\begin{enumerate}
\item
Can two languages given by a DFA be separated by a language specified by a counter-free DFA?
\item
Can two languages, each recognised by a finite semigroup, be separated by a language recognised by a finite aperiodic 
semigroup?
\item
Can two $\MSOo$-definable languages be separated by an $\FOo$-definable language? 
(In other words, can two $\MSO$-sentences in the signature $\sigma_\abc$ 
be separated by some $\FO$-sentence over finite strict linear orders?)
\end{enumerate}
As we shall see in the next section, this question also corresponds to the IEP in linear temporal logic $\LTL$ over finite timelines. 

The following result was first obtained with the help of purely algebraic methods:

\begin{theorem}[\cite{henkell1,almeida99}]\label{thm:FOsep}
It is decidable whether two regular languages are separable by an $\FOo$-definable language. 
\end{theorem}


A more general result of \cite{DBLP:journals/corr/PlaceZ14} (obtained using a mix of algebraic and combinatorial techniques) shows the decidability of the separation (and the more general \emph{covering}) problem for 
regular languages by any class $\mathcal{C}'$ of regular languages satisfying certain closure conditions.
Here, we present the algorithm from \cite{DBLP:journals/corr/PlaceZ14} for the case when $\mathcal{C}'$ is the class of all $\FOo$-definable languages.

Observe first that if $\sg_1,\sg_2$ are semigroups recognising regular languages $\lang_1,\lang_2$ with 
respective homomorphisms $\alpha_1,\alpha_2$, then the direct product $\sg_1\mathop{\times}\sg_2$ recognises both
$\lang_1$ and $\lang_2$ with the homomorphism $\alpha\colon\abc^+\to \sg_1\mathop{\times}\sg_2$ defined by taking
$\alpha(w)=\bigl(\alpha_1(w),\alpha_2(w)\bigr)$.
So we may assume that we are given two regular languages $\lang_1,\lang_2$
recognised by the same semigroup $\sg$, homomorphism  $\alpha\colon\abc^+\to \sg$, and respective sets $F_1,F_2\subseteq S$.
Given two subsets $T,T'$ of $S$, we set $T\sgo T'=\{t\sgo t'\mid t\in T,\ t'\in T'\}$.
For $T\in\powerset{S}$ and $0<n<\omega$, we write $T^n$ for $\underbrace{T\sgo \dots \sgo T}_n$.
We let $\satS$ be the smallest subset of $\powerset{S}$ that contains the singletons $\{\alpha(w)\}$, $w\in\abc^+$, and is closed under the following operations:
\begin{description}
\item[\quad C1]
if $T\in\satS$, then $T'\in\satS$ for all $T'\subseteq T$;

\item[\quad C2]
if $T,T'\in\satS$, then $T\sgo T'\in\satS$;

\item[\quad C3]
if $T\in\satS$, then $T^{\ino{\sg}}\cup T^{\ino{\sg}+1}\in\satS$.
\end{description}
It is shown in~\cite{DBLP:journals/corr/PlaceZ14} that
\begin{equation}\label{sepalgcrit}
\mbox{$\lang_1$ and $\lang_2$ are 
separable by an $\FOo$-definable language 
iff $\{t_1,t_2\}\notin\satS$ for all $t_1\in F_1$,
$t_2\in F_2$.}
\end{equation}
Thus, we obtain the following: 
\begin{theorem}[\cite{DBLP:journals/corr/PlaceZ14}]\label{co:sepexptime}
Deciding $\FOo$-separability of two regular languages is in $\ExpTime$ in the size of their recognising semigroups, and so in $2\ExpTime$ in the size of the NFAs
specifying the given languages.
\end{theorem}

By Theorem~\ref{co:defcom}, this problem is $\PSpace$-hard in the size of the DFAs.
We illustrate the algorithm by two examples.

\begin{example}\label{ex:yes}
It is well known (see, e.g.,~\cite{DBLP:books/sp/Libkin04}) that
\begin{equation}\label{parity}
\mbox{there is no $\FOo$-sentence 
	expressing parity of the length of a finite strict linear order.}
	\end{equation}
	Now, let $\lang_1=(abab)^+$ and $\lang_2=(baba)^+$. An easy adaptation of the usual Ehrenfeucht-Fra\"\i ss\'e argument for \eqref{parity}
	shows that these languages are not
	$\FOo$-definable. On the other hand, the language $\lang=(ab)^+$ clearly separates them, and it is not hard to see that
	$\lang$ is $\FOo$-definable: it is the set of all words that start with an $a$, end with a $b$, and do not contain $aa$ or $bb$ as
	a subword.
	
	As concerns the algorithm detecting this,
	observe first that both languages are recognised by the transition semigroup $\sg$ of the DFA depicted below, with $4_a$ as an accepting state for $\lang_1$, $4_b$ as an accepting state for  $\lang_2$, and the missing transitions going to a `sink' state $\bot$.\\
	
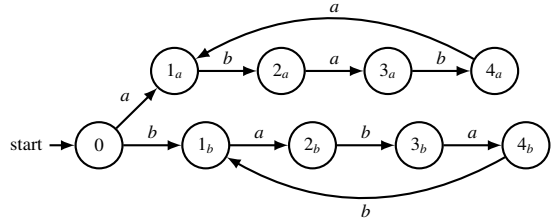
\centerline{
\begin{tikzpicture}[->,thick,scale=0.7,node distance=2cm, transform shape]
	\node[state, initial] (0) {$0$};
	\node[state, above right  of  =0] (1a) {$1_a$};
	\node[state, right  of  =1a] (2a) {$2_a$};
	\node[state, right  of  =2a] (3a) {$3_a$};
	\node[state, right  of  =3a] (4a) {$4_a$};
	\node[state, right  of  =0] (1b) {$1_b$};
	\node[state, right  of  =1b] (2b) {$2_b$};
	\node[state, right  of  =2b] (3b) {$3_b$};
	\node[state, right  of  =3b] (4b) {$4_b$};
	\draw (0) edge[above left] node{$a$} (1a)
	(1a) edge[above] node{$b$} (2a)
	(2a) edge[above] node{$a$} (3a)
	(3a) edge[above] node{$b$} (4a)
	(4a) edge[bend right=30,above] node{$a$} (1a)
	(0) edge[above] node{$b$} (1b)
	(1b) edge[above] node{$a$} (2b)
	(2b) edge[above] node{$b$} (3b)
	(3b) edge[above] node{$a$} (4b)
	(4b) edge[bend left=30,below] node{$b$} (1b)
	;
	\end{tikzpicture}}
	The recognising homomorphism $\alpha$ maps any $w\in\abc^+$ to the state-to-state transition function $\delta_w$.
	It is easy to see that
	$
	S=\{\delta_a,\delta_b,\delta_{aa},\delta_{ab},\delta_{ba},\delta_{aba},\delta_{bab},\delta_{abab},\delta_{baba}\},
	$
	with idempotent elements $\delta_{abab}$, $\delta_{baba}$ and $\delta_{aa}$ (with $\delta_{aa}$ being the constant $\bot$-function), and $\ino{\sg}=2$. Further, $\lang_1=\alpha^{-1}(\{\delta_{abab}\})$ and $\lang_2=\alpha^{-1}(\{\delta_{baba}\})$. Also,
	it is not hard to see that the only non-singleton elements of $\satS$
	are $\{\delta_{abab},\delta_{ab}\}$, $\{\delta_{a},\delta_{aba}\}$, $\{\delta_{baba},\delta_{ba}\}$, and $\{\delta_{b},\delta_{bab}\}$. 
	Therefore, $\{\delta_{abab},\delta_{baba}\}\notin\satS$, and so $\lang_1$ and $\lang_2$ are $\FOo$-separable by
	the criterion \eqref{sepalgcrit}. \lipicsEnd
\end{example}

The next example is borrowed from \cite{DBLP:journals/corr/PlaceZ14}:

\begin{example}\label{ex:no}
Let $\lang_1=\bigl(b(aa)^\ast b(aa)^\ast a\bigr)^+$ and 
$\lang_2=\bigl(b(aa)^\ast b(aa)^\ast a\bigr)^\ast b(aa)^\ast$ (where $\ast$ is the `real' Kleene star).
Again, an easy adaptation of the Ehrenfeucht-Fra\"\i ss\'e argument proving \eqref{parity} shows
that $\lang_1$ and $\lang_2$ are not $\FOo$-separable. 
To explain the algorithm detecting this, observe first that both languages are recognised by the transition semigroup $\sg$ of the DFA below, with $4$ as an accepting state for $\lang_1$, $2$ as an accepting state for  $\lang_2$, and the missing transitions going to a `sink' state $\bot$.\\
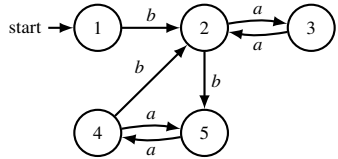
\centerline{
\begin{tikzpicture}[->,thick,scale=0.7,node distance=2cm, transform shape]
	\node[state, initial] (1) {$1$};
	\node[state, right  of  =1] (2) {$2$};
	\node[state, right  of  =2] (3) {$3$};
	\node[state, below of  =1] (4) {$4$};
	\node[state, right  of  =4] (5) {$5$};
	\draw (1) edge[above] node{$b$} (2)
	(2) edge[right] node{$b$} (5)
	(4) edge[above left] node{$b$} (2)
	(2) edge[bend left=10,above] node{$a$} (3)
	(3) edge[bend left=10,below] node{$a$} (2)
	(4) edge[bend left=10,above] node{$a$} (5)
	(5) edge[bend left=10,below] node{$a$} (4)
	;
	\end{tikzpicture}}
	The recognising homomorphism $\alpha$ maps any $w\in\abc^+$ to the state-to-state transition function $\delta_w$.
	It is easy to see that $\lang_1=\alpha^{-1}(\{\delta_{b^2a}\})$ and $\lang_2=\alpha^{-1}(\{\delta_b,\delta_{b^2ab}\})$.
	We claim that $\{\delta_{b^2a},\delta_{b^2ab}\}$ in $\satS$, and so $\lang_1$ and $\lang_2$ are not $\FOo$-separable by
	the criterion \eqref{sepalgcrit}.
	Indeed, observe that $\ino{\sg}=3$, so $\{\delta_a\}^{\ino{\sg}}=\{\delta_{a}\}$ and $\{\delta_a\}^{\ino{\sg}+1}=\delta_{a^2}$.
	Thus, $\{\delta_a,\delta_{a^2}\}\in\satS$ by C3. Then, by C2, we obtain that 
	$\{\delta_a,\delta_{a^2}\}\sgo\{\delta_b\}=\{\delta_{ab},\delta_{a^2b}\}\in\satS$.
	Let $X=\{\delta_{ab},\delta_{a^2b}\}$.
	As $\delta_{a^2b}\sgo\delta_{ab}\sgo\delta_{a^2b}=\delta_{bab^2}$ and
	$\delta_{a^2b}\sgo\delta_{ab}\sgo\delta_{a^2b}\sgo\delta_{ab}=\delta_{bab}$,
	we have $\delta_{bab^2},\delta_{bab}\in X^{\ino{\sg}}\cup X^{\ino{\sg}+1}$,
	and so $\{\delta_{bab^2},\delta_{bab}\}\in\satS$ by C3 and C1.
	And as $\delta_b\sgo\delta_{bab^2}\sgo\delta_a=\delta_{b^2a}$ and
	$\delta_b\sgo\delta_{bab}\sgo\delta_{aa}=\delta_{b^2ab}$, 
	we obtain that $\{\delta_{b^2a},\delta_{b^2ab}\}\in\satS$ by C2 and C1, as required. \lipicsEnd
\end{example}

{\bf Constructing a possible separator.}
The following is obtained by generalising Wilke's~\cite{DBLP:conf/stacs/Wilke99} construction of $\FO$-definitions for
$\FOo$-definable languages:

\begin{theorem}[\cite{DBLP:journals/corr/PlaceZ14}]\label{thm:sepsize}
Suppose $\lang_1,\lang_2\subseteq\abc^+$ are $\FOo$-separable languages that are recognised by a finite semigroup $\sg$.
Then one can construct an $\FO$-sentence $\varphi$ in the signature $\sigma_\abc$ such that
$\lang=\{w\in\Sigma^+\mid w\models\varphi\}$ separates $\lang_1$ and $\lang_2$, and the quantifier depth of $\varphi$ is 
linear in $|\abc|$ and
exponential in the size of $\sg$ \textup{(}and so double-exponential in the size of the DFAs specifying $\lang_1$ and $\lang_2$\textup{)}.
\end{theorem}

%
%

{\bf Languages over infinite words.}
According to \cite{DBLP:journals/corr/PlaceZ14}, the analogues of 
Theorems~\ref{thm:FOsep}, \ref{co:sepexptime}, and \ref{thm:sepsize} also hold for $\omega$-languages.

\smallskip


{\bf Definability and separability of regular languages by `simpler' languages.}
Several natural hierarchies of star-free languages have been introduced within each of the automata-theoretic, algebraic and logical approaches.
There has been a lot of work on the connections among these and, using these connections, on deciding the corresponding definability and separation problems; see, e.g.,~\cite{DBLP:journals/mst/PlaceZ19} for a relatively recent survey on definability 
and~\cite{PZ2025navigationalhierarchiesregularlanguages} for recent results on separation.

{\bf Definability and separability of non-regular languages.}
The definability and separability problems for  
languages specified by devices more powerful than finite automata 
(such as automata with counters, Parikh automata, vector addition systems with states, etc.)
have been mostly studied for the case when the
defining/separating language is regular. Even these questions attracted much less attention than defining/separating regular 
languages, possibly because of the lack of algebraic methods and the scarcity of positive results. 
A famous early negative result is that separability of context-free languages by a regular one is undecidable~\cite{DBLP:journals/siamcomp/SzymanskiW76}.
As most studied non-regular language-classes are not
closed under complementation, definability for these cannot be reduced to separation. In fact, there are cases when the
separation problem is decidable, while the definability problem is undecidable; see, e.g.,~\cite{DBLP:journals/lmcs/CzerwinskiL19} for examples and a wider context.

%
%



\section{Back Again: Interpolant Existence in Linear Temporal Logic}\label{sec:LTL}

In this final section, we use the decidability of $\FOo$-separability of regular languages (cf.\ Theorems~\ref{thm:buchi},
\ref{thm:FO} and \ref{thm:FOsep})
to decide the IEP for propositional linear temporal logic $\LTL$~\cite{GabEtAl03,DBLP:books/cu/Demri2016} over finite timelines, which is known to lack the CIP~\cite{DBLP:journals/jancl/Maksimova99} (see Example~\ref{ex:ltlnocip} below).

First, we briefly remind the reader of the relevant definitions. 
$\LTL$-\emph{formulas} are built from a countably-infinite set $p_0, p_1, \dots$ of propositional variables using the Booleans and temporal operators $\nxt$ (at the next moment), $\Diamond$ (some time in the future) and $\U$ (until). 
We interpret them in \emph{temporal models} that are Kripke models 
of the form
\[
\mathfrak{M} = \big(\{0,\dots,\ell\},<, \V_{\mathfrak M} \big) ,
\]
where $\ell < \omega$ and $<$ is the standard strict linear order.
For any set $\varrho$ of propositional variables,
we call a temporal model $\mathfrak{M}$ a $\varrho$-\emph{model} if $\V_{\mathfrak M}(p_i)=\emptyset$, for any $p_i\notin\varrho$.
The \emph{truth-relation} in $\mathfrak{M}$ (under the \emph{strict semantics} of $\Diamond$ and $\U$) is defined by taking, for every $n  \le \ell$  
and every propositional variable $p_i$,
\begin{itemize}
\item $\mathfrak{M},n \models p_i$ iff $n \in\V_{\mathfrak M}(p_i)$, 
\item $\mathfrak{M},n \models \nxt \varphi$ iff $n < \ell$ and $\mathfrak{M},n+1 \models \varphi$, 

\item $\mathfrak{M},n \models \Diamond \varphi$ iff $\mathfrak{M},m  \models \varphi$, for some $m$ with $n<m\le \ell$, 

\item $\mathfrak{M},n \models \varphi \U \psi$ iff there is $m$ such that  $n < m \le \ell$, $\mathfrak{M},m \models \varphi$ and $\mathfrak{M},k \models \psi$ for all $k$ with $n < k < m$.
\end{itemize}
Observe that temporal models can be regarded as $\stemp$-structures in the \FO-signature $\stemp$ that consists of a binary predicate symbol $<$ and a unary predicate symbol $P_i$, for each propositional variable $p_i$.
Thus, similarly to the \emph{standard translation} for modal logic,
the truth-conditions above can be expressed in the $\FO$-language in $\stemp$.
Moreover, there is a translation backwards, as the following generalisation of Kamp's Theorem~\cite{phd-kamp} by Gabbay, Pnueli, Shelah and Stavi~\cite{DBLP:conf/popl/GabbayPSS80} states:

\begin{theorem}[\cite{phd-kamp,DBLP:conf/popl/GabbayPSS80}]\label{thm:Kamp}
%
%
For every $\FO$-formula $\vartheta(x)$ in the signature $\stemp$ having one free variable $x$, there is an $\LTL$-formula
$\varphi_\vartheta$ such that $\mathfrak M\models\vartheta(0)$ iff $\mathfrak M,0\models\varphi_\vartheta$, for all temporal models $\mathfrak M$. 
\end{theorem}

Note that $\varphi_{\vartheta}$ is of non-elementary size in the size of $\vartheta$ in the worst case~\cite{Stockmeyer74}.  
If we consider the symbols of an alphabet $\abc$ as propositional variables, then any word $w\in\abc^+$ can be regarded as 
a temporal model. We say that a language $\lang\subseteq\abc^+$ is $\LTL$-\emph{definable} if there is an $\LTL$-formula $\varphi$ such that 
$\lang=\{w\in\abc^+\mid w,0\models\varphi\}$.

\begin{corollary}\label{co:LTLFO}
A language is $\FOo$-definable iff it is $\LTL$-definable.
\end{corollary}

On the other hand, Theorems~\ref{thm:buchi} and \ref{thm:Kamp} suggest a characterisation of regular languages in terms of existentially quantified $\LTL$-formulas. More precisely, define $\exists \LTL$ as the language that comprises $\LTL$-formulas possibly prefixed by quantifiers $\exists q$ over propositional variables $q$ that are interpreted by taking $\mathfrak{M},n \models \exists q \, \varphi$ iff $\mathfrak{M}',n \models \varphi$, for some $\mathfrak{M}'$ that may only differ from $\mathfrak{M}$ on the interpretation of $q$. 
Observe that
the truth-condition for $\exists q_1\dots\exists q_n\varphi$ can be expressed by an $\MSO$-formula of the form $\exists Q_1\dots\exists Q_n\,\psi_\varphi$, where $\psi_\varphi$ is the standard translation (in the signature $\stemp$) for $\varphi$, and each $Q_i$ is the unary predicate symbol corresponding to $q_i$. 
Moreover, the following is a consequence of Theorems~\ref{thm:buchi} and \ref{thm:Kamp}:


\begin{theorem}\label{thm:quantified}	
\subtheorem{thm:msotoltl}
For every $\MSO$-formula 
$\vartheta(x)$ 
in the signature $\stemp$ having one free first-order variable $x$, there is a $\exists\LTL$-formula
$\varphi_\vartheta$ such that $\mathfrak M\models\vartheta(0)$ iff $\mathfrak M,0\models\varphi_\vartheta$, 
for all temporal models $\mathfrak M$. It follows that:

\subtheorem{thm:regLTL}
A language is regular iff it is $\exists \LTL$-definable.
\end{theorem}

Given $\LTL$ or $\exists\LTL$-formulas $\varphi,\psi$, we write $\varphi\models_{\LTL}\psi$ whenever 
$\mathfrak{M},0\models\varphi$ implies $\mathfrak{M},0\models\psi$, for every $\mathfrak{M}$.
(Note that this is equivalent to the definition of $\models_L$ for modal logics given in Section~\ref{Sec:2}.)

\begin{example}\label{ex:ltlnocip}
We show two $\LTL$-formulas that do not have a Craig interpolant in $\LTL$. Consider
\[
\varphi = p\land\Box^+(p\land\nxt\top\leftrightarrow\nxt\neg p)\land\Diamond(\neg p\land\neg\nxt\top),\quad \  
\psi= q\land\Box^+(q\land\nxt\top\leftrightarrow\nxt\neg q)\to\Diamond(\neg q\land\neg\nxt\top).
\]
%
Then, for every $\mathfrak M$, we have
\[
\mbox{$\mathfrak M,0\models\varphi$\ \  iff\ \   $p$ holds in $\mathfrak M$ at even numbers and the number of points in $\mathfrak M$ is even},
\]
%
and so $\varphi\models_{\LTL}\psi$. Any interpolant $\chi$ for $\varphi,\psi$ in $\LTL$ should be variable-free and express that  the number of points in an $\emptyset$-model is even;
the standard translation of such a $\chi$ should be an $\FOo$-sentence 
expressing parity of the length
of a finite strict linear order, which is impossible by \eqref{parity}.
\lipicsEnd
\end{example}

So far, the decidability status of the IEP for $\LTL$ has been open, with the methods of analysing the problem discussed in Section~\ref{Sec:3} being non-applicable in this case. The next theorem settles this open question by employing Theorem~\ref{thm:Kamp} to reduce---for the price of a multi-exponential blow-up---the IEP for $\LTL$ to $\FOo$-separability of regular languages and then using Theorems~\ref{thm:FO}, \ref{thm:FOsep} and \ref{co:sepexptime}.

\begin{theorem}\label{thm:ltliep}
The IEP for $\LTL$ is decidable in $3\ExpTime$,
being \PSpace-hard.
\end{theorem}
\begin{proof}
Suppose we are given $\LTL$-formulas $\varphi$ and $\psi$ whose shared propositional variables comprise a set $\varrho$. Our aim is to decide whether there exists an $\LTL$-formula $\chi$ built from variables in $\varrho$ such that
$\varphi\models_{\LTL}\chi$ and  $\chi\models_{\LTL}\psi$. 
To this end, we take the alphabet $\abcr=2^\varrho$ and 
observe that any $\varrho$-model $\mathfrak{M}$ can be regarded as a 
word $w\in\abcr^+$, where $w = \ap_0\dots \ap_\ell$, for $\ap_i=\{p\in\varrho\mid i\in\V_{\mathfrak{M}}(p)\}$.
Let $q_1, \dots,q_m$ and $r_1, \dots, r_k$ be the respective lists of variables in $\varphi$ and $\psi$ that are not in $\varrho$.
We define two languages $\lang_\varphi,\lang_{\neg\psi}\subseteq\abcr^+$ by taking
\[
\lang_\varphi=\{\varrho\mbox{-model }\mathfrak{M}\mid \mathfrak{M},0\models\exists q_1 \dots \exists q_m \, \varphi\}, \qquad 
\lang_{\neg\psi}=\{\varrho\mbox{-model }\mathfrak{M}\mid \mathfrak{M},0\models\exists r_1 \dots \exists r_k \, \neg\psi\}.
\]
By Theorem~\ref{thm:regLTL}, $\lang_\varphi$ and $\lang_{\neg\psi}$ are both regular languages. We claim that 
\begin{equation}\label{iepred}
\mbox{ there is a Craig interpolant for $\varphi,\psi$ in $\LTL$\quad iff\quad $\lang_\varphi$ and $\lang_{\neg\psi}$ are $\FOo$-separable.}
\end{equation}
Observe that if $\chi$ is any $\LTL$-formula not containing the variables $ q_1, \dots,q_m,r_1,\dots,r_k$, then
\begin{equation}\label{qfequiv}
\mbox{$\varphi\models_{\LTL}\chi$\ \ iff\ \  $\exists q_1 \dots \exists q_m\, \varphi\models_{\LTL}\chi$,\qquad and\qquad
	$\chi\models_{\LTL}\psi$\ \  iff\ \  $\chi\models_{\LTL}\neg\exists r_1 \dots \exists r_k \neg\psi$.}
	\end{equation}
	To prove \eqref{iepred},
	suppose first that $\chi$ is an $\LTL$-formula built up from variables in $\varrho$ such that
	$\varphi\models_{\LTL}\chi$ and  $\chi\models_{\LTL}\psi$. 
	We define a language $\lang_\chi\subseteq\abcr^+$ by taking
	$\lang_\chi=\{\varrho\mbox{-model }\mathfrak{M}\mid \mathfrak{M},0\models\chi\}$.
	By \eqref{qfequiv}, $\exists q_1 \dots \exists q_m\, \varphi\models_{\LTL}\chi$ and 
	$\chi\models_{\LTL}\neg\exists r_1 \dots \exists r_k \neg\psi$, and so 
	$\lang_\varphi\subseteq\lang_\chi$ and $\lang_{\neg\psi}\cap\lang_\chi=\emptyset$.
	Thus, $\lang_\chi$ separates $\lang_\varphi$ and $\lang_{\neg\psi}$.
	Moreover, $\lang_\chi$  is $\FOo$-definable. Indeed, the standard translation of $\chi$ is an equivalent $\FO$-formula $\vartheta_\chi(x)$ in the signature $\stemp$ such that every predicate symbol $P$ in $\vartheta_\chi$ corresponds to some $p\in\varrho$.
	By replacing each subformula in $\vartheta_\chi$ of the form
	$P(y)$, $p\in\varrho$, with $\bigvee_{\ap\in\abcr,\,p\in\ap}\ap(y)$, we obtain an $\FO$-formula $\vartheta'_\chi(x)$ in the signature
	$\sigma_{\abcr}$ that is equivalent to $\vartheta_\chi(x)$ over $\varrho$-models.
	%

	Conversely, suppose that $\lang_\varphi$ and $\lang_{\neg\psi}$ are separated by an $\FOo$-definable language $\lang$, and $\vartheta$ is a $\FO$-sentence in the signature $\sigma_{\abcr}$ such that
	$\lang= \{\varrho\mbox{-model }\mathfrak{M}\mid \mathfrak{M}\models\vartheta\}$. 
	By replacing each subformula of $\vartheta$ of the form $\ap(y)$ by $\bigwedge_{p\in\ap}P(y)$,
	we obtain an $\FO$-sentence $\vartheta^-$ in the signature $\stemp$, containing only unary predicates $P$ for $p\in\varrho$, that is equivalent to $\vartheta$ over $\varrho$-models.
	By Theorem~\ref{thm:Kamp}, there is some $\LTL$-formula $\chi_{\vartheta^-}$ containing only propositional variables from $\varrho$ and such that $\mathfrak{M}\models \vartheta^-$ iff $\mathfrak{M},0 \models \chi_{\vartheta^-}$, for all $\varrho$-models $\mathfrak{M}$.
	Now, it follows from  $\lang_\varphi\subseteq\lang$ that $\exists q_1 \dots \exists q_m\, \varphi\models_{\LTL}\chi_{\vartheta^-}$, and so
	$\varphi\models_{\LTL}\chi_{\vartheta^-}$ by \eqref{qfequiv}. Similarly,
	it follows from $\lang_{\neg\psi}\cap\lang=\emptyset$ that $\chi_{\vartheta^-}\models_{\LTL}\neg\exists r_1 \dots \exists r_k\neg\psi$,
	and so $\chi_{\vartheta^-}\models_{\LTL}\psi$ by \eqref{qfequiv}. Therefore, $\chi_{\vartheta^-}$ is an interpolant for $\varphi,\psi$ in $\LTL$.
	
	As for the computational complexity of deciding the IEP for \LTL, the above reduction gives a 
	$3\ExpTime$ upper bound:
	the size of an NFA $\mathfrak A$ specifying the language defined by an $\LTL$-formula $\varphi$
	is exponential in the size of $\varphi$~\cite{DBLP:conf/focs/WolperVS83}, 
	the NFA specifying $\lang_\varphi$ is a projection of $\mathfrak A$~\cite{DBLP:journals/tcs/SistlaVW87},
	and so has size linear in the size of $\mathfrak A$; 
	and then two more exponentials come from Theorem~\ref{co:sepexptime}.
\end{proof}

It is to be noted that, conversely, $\FOo$-separability of regular languages can also be reduced to the IEP for $\LTL$ by using Corollary~\ref{co:LTLFO} and Theorem~\ref{thm:quantified}. 
The exact computational complexities of both of these decision problems  remain open.


\smallskip


{\bf $\LTL$ over $(\omega,<)$.} $\LTL$-formulas can also be interpreted in infinite temporal models of the form
$(\omega,<,\V)$. We denote the resulting logic by $\LTLinf$.
This kind of temporal models can also be regarded as $\stemp$-structures of the \FO-signature $\stemp$ above.
As discussed above, the analogues of Theorems~\ref{thm:FO}, \ref{thm:FOsep}, and \ref{co:sepexptime}  hold for $\omega$-languages.
Moreover, Theorem~\ref{thm:Kamp} also holds in this setting, and so a proof similar to that of Theorem~\ref{thm:ltliep} shows that the IEP for $\LTLinf$ is also decidable, with its exact complexity remaining also open.

{\bf Fragments of $\LTL$.}
Denote by $\LTLd$ the fragment of $\LTL$ that only allows $\Diamond$ as a temporal operator, and by $\LTLnd$ the fragment, in which both $\nxt$ and $\Diamond$ are allowed. Neither of them enjoys the CIP. Using the method outlined in Section~\ref{Sec:3}, it is shown in~\cite{DBLP:journals/corr/abs-2312-05929} that the IEP for $\LTLd$ over both finite and infinite linear timelines is \textsc{coNP}-complete, that is, as complex as validity~\cite{Sistla&Clarke85}. 
The decidability of the IEP for $\LTLnd$ over finite and infinite linear timelines can be established via the `detour' to language separation described earlier in this section, using the following two facts. First,  every $\FO^2[<,\textit{suc}]$-formula with one free variable is equivalent to an $\LTLnd$-formula~\cite{DBLP:journals/iandc/EtessamiVW02,DBLP:conf/popl/GabbayPSS80,DBLP:journals/eatcs/Markey03}. And second, it is decidable whether two regular languages are separable by an $\FO^2[<,\textit{suc}]$-definable language~\cite{DBLP:conf/stacs/PlaceZ15}. The exact complexity of the IEP for $\LTLnd$ is open.


It is also of interest to note that quantified $\LTL$ (with both existential and universal quantification over propositional variables)  
is the smallest `natural' logic containing $\LTLd$ and having the CIP~\cite{DBLP:conf/csl/GheerbrantC09}.

\section{Research Directions and Open Problems}\label{conclusion}

In this chapter, we have provided an overview of recent research into variations and generalisations of Craig interpolation, the common feature of which is that the existence of an interpolant does not reduce to entailment, either because the logic in question does not enjoy the CIP, or because the use of logical symbols in the interpolant is restricted. Numerous research questions remain to be explored. We conclude this chapter by discussing a few of them: 

\begin{description}
\item[The role of proof theory.] The methods presented in this chapter have been mainly model-theoretic or automata-based. In fact, proof theory has not yet been applied to decide interpolant existence or construct interpolants for logics without the CIP or for separation. 
Developing proof-theoretic methods would be of great interest, in particular since very little is known about the actual construction of interpolants and separators, and proof-theoretic methods are a promising tool to address this problem as demonstrated by their wide applications for logics with the CIP, both in theory and practice.

\item[Non-classical logics.] The focus of this chapter has been on modal and first-order logics, and their extensions with second-order quantifiers and fixed points. We have not discussed important families of non-classical logics such as superintuitionistic, many-valued, substructurual, or relevant logics. To the best of our knowledge the variations of Craig interpolation investigated in this chapter have not yet been considered for these logics, despite the fact that many non-classical logics do not enjoy the CIP and that it appears natural to compare the expressive power of different non-classical logics by studying separation. 			

\item[General results.] Table~\ref{TableCraig} shows first results on computational complexity of the IEP, but there are large gaps. For instance, except for $\SFi$, hardly anything is known about the IEP or separation in first-order modal and temporal logic, despite the fact that under constant domain semantics they all lack the CIP. Similarly, are there general results about the IEP for logics containing $\SF$ or $\KF$, extending the result for logics above $\KFT$? Again, this would be of interest since at most 36 extensions of $\SF$ have the CIP. Finally, although a few results have been discussed, the decidability status for interpolant existence for many decidable fragments of \FO{} remains open. An example is the two-variable fragment of \FO{} with counting.   
\end{description}

\section*{Acknowledgments}
\addcontentsline{toc}{section}{Acknowledgments}

The authors thank Jean Christoph Jung and J{\c e}drzej Ko{\l}odziejski for
their valuable comments and suggestions.

\providecommand{\noopsort}[1]{}\providecommand{\noopsort}[1]{}

%
%
%
%
%
%
%
\end{document}